\documentclass[aps,pra,showpacs,twocolumn,amsmath,superscriptaddress]{revtex4-2}
\usepackage{color}
\usepackage{graphicx,epsfig,subfigure,dsfont,amssymb,amsmath,amsthm,amsfonts,amsbsy,mathrsfs,amscd}
\usepackage{bm}
\usepackage{bbm}
\usepackage{contour}
\usepackage{latexsym}
\usepackage{amssymb} 
\usepackage{amsthm} 
\usepackage{times} 
\usepackage{graphicx}

\newcommand{\aop}{Ann. Phys.~}

\newcommand{\jpa}{J. Phys. A: Math. Theor.~}

\newcommand{\njp}{New. J. Phys.~}

\newcommand{\pla}{Phys. Lett. A~}

\newcommand{\tinyspace}{\mspace{1mu}}

\newcommand{\op}[1]{\operatorname{#1}}

\newcommand{\norm}[1]{\left\lVert\tinyspace #1 \tinyspace\right\rVert}

\renewcommand{\det}{\operatorname{det}}
\renewcommand{\t}{{\scriptscriptstyle\mathsf{T}}}

\newcommand{\rank}{\op{rank}}

\def\sl{\mathfrak{sl}}

\def \diag {\mathrm{diag}}

\def\I{\mathbb{1}}

\def\I{\mathbbm{1}}

\def\1{\mathbf{1}}
\newcommand{\proj}[1]{| #1\rangle\!\langle #1 |}

\def\ot{\otimes}

\newcommand{\out}[2]{| #1\rangle\langle #2 |}
\newcommand{\Inner}[2]{\left\langle #1 , #2\right\rangle}
\newcommand{\Innerm}[3]{\left\langle #1 \left| #2 \right| #3 \right\rangle}

\newcommand{\Herm}{\mathrm{Herm}}

\newcommand{\pa}[1]{(#1)}
\newcommand{\Pa}[1]{\left(#1\right)}

\newcommand{\Br}[1]{\left[#1\right]}
\newcommand{\set}[1]{\{#1\}}
\newcommand{\Set}[1]{\left\{#1\right\}}

\newcommand{\bra}[1]{\langle#1|}

\newcommand{\ket}[1]{|#1\rangle}

\DeclareMathOperator{\trace}{Tr}
\newcommand{\ptr}[2]{\trace_{#1}\pa{#2}}
\newcommand{\Ptr}[2]{\trace_{#1}\Pa{#2}}

\newcommand{\Tr}[1]{\Ptr{}{#1}}

\newcommand{\Abs}[1]{\left|\tinyspace#1\tinyspace\right|}

\def\cA{\mathcal{A}}\def\cB{\mathcal{B}}\def\cE{\mathcal{E}}
\def\cF{\mathcal{F}}\def\cH{\mathcal{H}}
\def\cK{\mathcal{K}}\def\cO{\mathcal{O}}
\def\cR{\mathcal{R}}\def\cS{\mathcal{S}}
\def\cW{\mathcal{W}}

\def\bbC{\mathbb{C}}

\def\bbR{\mathbb{R}}

\def\bsA{\boldsymbol{A}}\def\bsB{\boldsymbol{B}}\def\bsC{\boldsymbol{C}}
\def\bsF{\boldsymbol{F}}\def\bsH{\boldsymbol{H}}
\def\bsK{\boldsymbol{K}}\def\bsL{\boldsymbol{L}}\def\bsM{\boldsymbol{M}}\def\bsO{\boldsymbol{O}}
\def\bsS{\boldsymbol{S}}\def\bsT{\boldsymbol{T}}
\def\bsU{\boldsymbol{U}}\def\bsV{\boldsymbol{V}}\def\bsW{\boldsymbol{W}}\def\bsX{\boldsymbol{X}}\def\bsY{\boldsymbol{Y}}
\def\bsZ{\boldsymbol{Z}}

\def\bsa{\boldsymbol{a}}\def\bsb{\boldsymbol{b}}

\def\bsx{\boldsymbol{x}}\def\bsy{\boldsymbol{y}}

\def\rD{\mathrm{D}}

\def\rO{\mathrm{O}}
\def\rS{\mathrm{S}}
\def\rU{\mathrm{U}}

\def\sC{\mathscr{C}}

\def\B{\textsf{B}}

\def\U{\textsf{U}}

\newtheorem{thrm}{Theorem}
\newtheorem{lem}{Lemma}
\newtheorem{prop}{Proposition}
\newtheorem{cor}{Corollary}
\theoremstyle{definition}
\newtheorem{definition}{Definition}

\newtheorem{exam}{Example}

\begin{document}


\title{Bargmann Invariants Are Insufficient for Complete Local-Unitary Orbit Discrimination}


\author{Lin Zhang}
\email{Contact author: godyalin@163.com}
\author{Bing Xie}
\email{E-mail: xiebingjiangxi2023@163.com}
\affiliation{School of Mathematical Sciences, Hangzhou Dianzi
University, Hangzhou 310018, China}


\date{\today}

\begin{abstract}
Bargmann invariants constructed from a bipartite density operator and its two lifted marginals are polynomial invariants of local-unitary conjugation. We determine the precise information encoded in these invariants. Whenever one subsystem is a qubit, the ordinary marginal-word family determines the full spectrum of the partial transpose and hence decides whether the state has the positive-partial-transpose property. In $2\ot2$ and $2\ot 3$ systems, this yields complete separability criteria. For the two-qubit system, a finite subfamily additionally separates local-unitary orbits, and a finite extension generates the polynomial invariant ring. These three tasks already diverge for qubit-qutrit states: we exhibit full-rank, locally maximally mixed states that agree on all ordinary marginal-word invariants yet have different operator-Schmidt ranks, together with a quartic correlation invariant that separates them. When both local dimensions are at least three, the analogous collapse on the locally maximally mixed sector produces isospectral pairs consisting of one separable state and one entangled state with negative partial transpose. The ordinary Bargmann algebra therefore coincides with the full local-unitary invariant ring if and only if both subsystems are qubits. The missing data are geometric: they encode the placement of global eigenspaces relative to the tensor-product decomposition.
\end{abstract}

\maketitle

\section{Introduction}

The classification of composite quantum states under local changes of basis is a fundamental problem in quantum information theory   \cite{Grassl1998,Makhlin2002,Kraus2009,Zhou2012,Gour2013,Turner2017}. Let $\rho_{AB},\sigma_{AB}\in\rD(\bbC^{d_A}\ot\bbC^{d_B})$, the set of all bipartite density operators acting on $\bbC^{d_A}\ot\bbC^{d_B}$. They are locally unitarily equivalent if there exist unitaries $\bsU_A \in \rU(d_A)$ and $\bsU_B \in \rU(d_B)$ such that
\begin{eqnarray*}
\sigma_{AB} = (\bsU_A \ot\bsU_B)\rho_{AB}(\bsU_A \ot\bsU_B)^\dagger.
\end{eqnarray*}
Local-unitary (LU) equivalence identifies states that differ only by independent choices of basis on the two subsystems. It therefore preserves all intrinsic bipartite properties, including separability, entanglement measures, marginal spectra, operator-Schmidt coefficients, and correlation-tensor invariants \cite{Maciazek2013,Viehmann2011,Monras2011,Wyderka2023}.

This relation is strictly finer than global unitary equivalence. Two finite-dimensional Hermitian operators are globally unitarily equivalent if and only if they have the same eigenvalues with the same multiplicities. A general  $\bsW \in \rU(d_A d_B)$ need not respect the tensor-product decomposition and may send a separable state to an entangled one, or conversely. Any complete local-unitary classification must therefore retain not only the eigenvalues of $\rho_{AB}$ but also the position of its eigenspaces relative to the underlying bipartite structure.

Polynomial invariants provide a natural language for the problem. Because the local unitary group $\rU(d_A)\times \rU(d_B)$ is compact, its polynomial invariants separate distinct local-unitary orbits. The complete invariant ring is consequently sufficient in principle, but an explicit finite generating or separating family is difficult to identify, and the complexity grows rapidly with the local dimensions. Trace invariants are especially attractive: they are basis independent, algebraically natural, and compatible with multivariate trace-estimation protocols \cite{Oszmaniec2024,Quek2024}. Their role is closely related to the classical invariant theory of matrix tuples \cite{Procesi1976}.

Bargmann invariants were originally introduced in the study of projective unitary transformations \cite{Bargmann1964,Chien2016} and later became central to the theory of geometric phases \cite{Mukunda2003,Berry1984,Aharonov1987}. Cyclic trace products also appear in Kirkwood–Dirac representations \cite{Kirkwood1933,Dirac1945,Schmid2024}, relational measurements, and quantum information protocols \cite{Buhrman2001}.  Recent work has connected Bargmann-type quantities with imaginarity, coherence, and local-unitary discrimination \cite{Fernandes2024,Li2025,Li2026,Wagner2025,Zhang2025PRA1,Zhang2025,Pratapsi2025,Xu2026,Zhang2026,Wang2026,Wang2026,Wagner2026a,Wagner2026b}.

The present paper isolates a particularly canonical family and determines its exact limitations. For a bipartite state $\rho_{AB}$ with marginals $\rho_A = \ptr{B}{\rho_{AB}}$ and $\rho_B = \ptr{A}{\rho_{AB}}$, define the operator triple
\begin{eqnarray*}
\bsX_0 := \rho_{AB},  \bsX_1:= \rho_A \ot  \I_B,  \bsX_2:=\I_A \ot \rho_B.
\end{eqnarray*}
For every finite word $\mathbf{w} = i_1 i_2 \dots i_m$ with $i_k \in \set{0,1,2}$, the associated ordinary marginal-word Bargmann invariant is
\begin{eqnarray}\label{eq:Bargmann}
B_{\mathbf{w}}(\rho) := \Tr{\bsX_{i_1} \bsX_{i_2} \dots \bsX_{i_m}}.
\end{eqnarray}
Under a local unitary transformation, all three operators $\bsX_k(k=0,1,2)$ undergo simultaneous conjugation by $\bsU_A\ot\bsU_B$, so every $B_{\mathbf{w}}$ is LU invariant. Cyclicity of the trace gives $B_{i_1i_2\cdots i_m}=B_{i_2\cdots i_mi_1}$; further polynomial relations may appear in fixed dimensions.

The adjective “ordinary” is essential. Our negative results concern precisely the trace-word algebra $\bbC[B_{\mathbf{w}}:\mathbf{w}\in\set{0,1,2}^*]$ generated by the canonical triple $(\bsX_0,\bsX_1,\bsX_2)$. They do not exclude the possibility that broader Bargmann-type or correlation invariants achieve complete LU classification.

We address three related but logically distinct questions.
\begin{itemize}
\item \emph{Entanglement discrimination.}  Does the family $\set{B_{\mathbf{w}}:\mathbf{w}\in\set{0,1,2}^*}$ decide whether a bipartite state $\rho_{AB}$ is separable? More modestly, does it determine whether the partial transpose is positive?
\item \emph{Local-unitary orbit discrimination.}  If two states $\rho$ and $\rho'$ satisfy $B_{\mathbf{w}}(\rho)=B_{\mathbf{w}}(\rho')$  for every word $\mathbf{w}$, must they be locally unitarily equivalent?
\item \emph{Polynomial invariant-ring generation.}  Does the algebra generated by all $B_{\mathbf{w}}$ coincide with the complete ring of polynomial LU invariants?
\end{itemize}
These tasks should not be conflated. Deciding positivity of the partial transpose generally requires less information than deciding separability. Separability discrimination is weaker than complete orbit discrimination: two inequivalent states may both be separable or both be entangled. A family that separates orbits need not generate every polynomial invariant. The natural informational hierarchy is therefore
\begin{eqnarray*}
&&\text{PPT testing}\\
&&~~~\preceq\text{separability decision}\\
&&~~~~~~\preceq\text{LU-orbit separation}\\
&&~~~~~~~~~\preceq\text{full invariant-ring generation},
\end{eqnarray*}
where $\preceq$ indicates that a task may require no more information than the task to its right. Some of these distinctions collapse in low dimension;  they become strict as soon as the local dimensions increase.

Our first observation is an identity specific to a qubit subsystem. For every state on $\bbC^2\ot\bbC^d$,
\begin{eqnarray*}
(\sigma_y\ot\I_B)\rho^\Gamma_{AB}(\sigma_y\ot\I_B)=\I_A\ot\rho_B - \rho_{AB},
\end{eqnarray*}
where $\sigma_y$ is the second Pauli matrix and $\Gamma=\t_A$ denotes partial transposition on the first subsystem $A$ throughout. Thus $\rho^\Gamma_{AB}$ is unitarily similar to $\bsX_2-\bsX_0$. Although $\bsX_0$ and $\bsX_2$ need not commute, every moment $\Tr{[\rho^\Gamma_{AB}]^k}$ expands as a finite linear combination of ordinary trace words. Newton’s identities then express every coefficient of the characteristic polynomial of $\rho^\Gamma_{AB}$ as a polynomial in these invariants. Consequently, for every qubit-qudit state the ordinary Bargmann invariants determine the complete partial-transpose spectrum and decide the positive partial-transpose (PPT) property.

In $2\ot2$ and $2\ot3$ systems, the Peres–Horodecki criterion makes PPT equivalent to separability \cite{Peres1996,Horodecki1996}. For two-qubit states, the partial transpose has at most one negative eigenvalue, so entanglement is detected by the sign of its determinant: a single explicit polynomial inequality in ordinary Bargmann invariants. Together with a previously established finite separating family of 18 LU words \cite{Zhang2025} and a finite generating set for the two-qubit polynomial invariant ring \cite{Zhang2026}, this produces a complete picture. The two-qubit case is exceptional in that the three questions above all have affirmative answers within the same traceword framework.
\begin{itemize}
\item ordinary Bargmann invariants decide two-qubit separability;
\item a finite family separates all two-qubit LU orbits;
\item a finite extension of that family generates the full two-qubit LU-invariant ring.
\end{itemize}

For qubit-qutrit states, ordinary Bargmann invariants continue to decide entanglement, but the determinant alone
is no longer sufficient. We derive the first six partial-transpose moments from the noncommutative expansion of $\Tr{[\rho^\Gamma_{AB}]^k}=\Tr{(\bsX_2-\bsX_0)^k}$, convert them into elementary symmetric polynomials, and obtain a necessary and sufficient separability criterion consisting of four polynomial inequalities. Separability discrimination, however, does
not imply orbit discrimination. We construct two fullrank qubit–qutrit states $\rho(\bsH)$ and $\rho(\bsK)$ with identical global spectra, maximally mixed marginals, and
\begin{eqnarray*}
B_{\mathbf{w}}(\rho(\bsH))=B_{\mathbf{w}}(\rho(\bsK))
\end{eqnarray*}
for every finite word $\mathbf{w}$, yet the two states are not LU equivalent: their traceless correlation operators have different operator-Schmidt ranks. The same pair shows that the ordinary Bargmann algebra is a proper subalgebra of the full polynomial LU-invariant ring. An explicit degree-four invariant $J_2$, built from the second elementary symmetric polynomial of the squared singular values of the correlation matrix, separates the two states.

The mechanism is structural. If both marginals are maximally mixed, then $\bsX_1$ and $\bsX_2$ are scalar multiples of the identity, and
\begin{eqnarray}\label{eq:genericexpression}
B_{\mathbf{w}}(\rho)=d^{-N_1(\mathbf{w})}_Ad^{-N_2(\mathbf{w})}_B\Tr{\rho^{N_0(\mathbf{w})}_{AB}},
\end{eqnarray}
where $N_k(\mathbf{w})$ counts occurrences of the symbol $k$ in $\mathbf{w}$. On this locally maximally mixed sector, the entire infinite family of ordinary Bargmann invariants collapses to the global spectral moments of $\rho_{AB}$ and loses all information about the tensor-product geometry of the eigenspaces.

When both local dimensions are at least three, the failure is stronger still. For every $n\geqslant3$, we construct two states on $\bbC^n\ot\bbC^n$ with the same global spectrum and maximally mixed marginals, one separable and one negative-partial-transpose (NPT) entangled. Eq.~\eqref{eq:genericexpression} then implies that every ordinary marginal-word invariant agrees on the pair. In these dimensions the full
infinite family fails even to decide entanglement.

The remainder of the paper is organized as follows. Section~\ref{sect:2} records global versus local unitary equivalence, defines the ordinary Bargmann invariants, and proves that they determine the partial-transpose spectrum of every qubit–qudit state. Section~\ref{sect:3} develops the complete two-qubit picture. Section~\ref{sect:4} treats the qubit–qutrit system: a complete separability criterion, explicit failure of orbit separation, and an invariant outside the trace-word algebra, together with a complementary counterexample whose marginals are nondegenerate. Section~\ref{sect:5} constructs isospectral separable–NPT pairs in
every $n\times n$ system with $n\geqslant3$. Section~\ref{sect:6} proves that the ordinary Bargmann algebra equals the full polynomial
LU-invariant ring if and only if $(d_A,d_B)=(2,2)$. Section~\ref{sect:7} collects conclusions and open problems. Appendix~\ref{app:A} records the positivity criterion used for partial transposes; Appendix~\ref{app:B} computes the quartic correlation invariant; Appendix~\ref{app:C} gives the dimension count on the locally maximally mixed slice.

\section{Preliminaries}\label{sect:2}

\subsection{ Global and local unitary equivalence}

Let $\rho_{AB}, \sigma_{AB} \in \rD(\bbC^{d_A} \ot\bbC^{d_B})$. Two states are \emph{globally unitarily equivalent} if $\sigma_{AB} = \bsW\rho_{AB}\bsW^\dagger$ for some unitary $\bsW \in \rU(d_A d_B)$, and \emph{locally unitarily equivalent} if $\bsW$ may be chosen of product form $\bsU_A \ot\bsU_B$. Global unitary equivalence is completely determined by the spectrum. Local unitary equivalence is strictly finer because it preserves the tensor-product decomposition.

In particular, local unitaries preserve separability and entanglement, the spectra of the reduced states, the operator-Schmidt coefficients, and all polynomial invariants of the local-unitary action. A general global unitary need not preserve any of these bipartite properties. States on a single global unitary orbit may therefore occupy distinct local-unitary orbits and may even differ in entanglement.

The operator-Schmidt rank of a bipartite operator $\bsZ$ is the smallest integer $r$ such that $\bsZ=\sum^r_{k=1}\bsA_k\ot\bsB_k$.  Local-unitary conjugation acts by invertible transformations on the two local operator spaces and therefore preserves this rank.

\subsection{Ordinary marginal-word Bargmann invariants}

Let $\Theta=\set{0,1,2}^*$ denote the set of all finite words on the alphabet $\set{0,1,2}$. The ordinary Bargmann map sends a
state to the family
\begin{eqnarray*}
\mathfrak{B}(\rho)=\Pa{B_{\mathbf{w}}(\rho)}_{\mathbf{w}\in\Theta}.
\end{eqnarray*}
If $\rho'_{AB}=(\bsU_A \ot\bsU_B)\rho_{AB}(\bsU_A \ot\bsU_B)^\dagger$, the transformed marginals are $\rho'_A=\bsU_A\rho_A\bsU^\dagger_A$ and $\rho'_B=\bsU_B\rho_B\bsU^\dagger_B$. Writing $\bsU=\bsU_A\ot\bsU_B$, one has $\bsX'_k=\bsU\bsX_k\bsU^\dagger$ for $k=0,1,2$, and therefore $B_{\mathbf{w}}(\rho')=B_{\mathbf{w}}(\rho)$ for every word $\mathbf{w}$.

Because $\bsX_0,\bsX_1,\bsX_2$ are Hermitian, the generalized
Specht criterion \cite{Horn2012} for matrix tuples implies that equality of all
ordinary trace words is equivalent to simultaneous unitary
equivalence of the triples. More precisely, for two states
$\rho_{AB}$ and $\rho'_{AB}$, $B_{\mathbf{w}}(\rho)=B_{\mathbf{w}}(\rho')$ for every $\mathbf{w}\in\Theta$
if and only if there exists a global unitary
$\bsW\in \rU(d_Ad_B)$ such that
$\bsX_k(\rho')=\bsW \bsX_k(\rho)\bsW^\dagger$ for $k=0,1,2$.
Here the unitary $W$ is global and need not be of product
form. This distinction is the source of the counterexamples
below.

The preceding relation implies that $\rho_A$ and $\rho'_A$
have the same spectrum, and likewise for $\rho_B$ and
$\rho'_B$. By the spectral theorem, there exist unitaries
$\bsL_X\in \rU(d_X)$, $X=A,B$, such that $\rho_X=\bsL_X\rho'_X \bsL_X^\dagger$.
Let
\begin{eqnarray*}
\bsV=\bsL_A\otimes \bsL_B,\quad
\bsU=\bsV\bsW,\quad
\widetilde{\rho}_{AB}=\bsV\rho'_{AB}\bsV^\dagger.
\end{eqnarray*}
Then
\begin{eqnarray*}
\widetilde{\rho}_{AB}=\bsU\rho_{AB}\bsU^\dagger,
\quad
\bsU\bsX_k\bsU^\dagger=\bsX_k\quad(k=1,2).
\end{eqnarray*}
and hence $[\bsU,\bsX_k]=0$ for $k=1,2$.
Denote by $\sC(\tau_A,\tau_B)$ the set of bipartite density operators with prescribed
marginals $\tau_A$ and $\tau_B$ \cite{Parthasarathy2005}. For any
$\tau,\tau'\in\sC(\tau_A,\tau_B)$, one has $B_{\mathbf{w}}(\tau)=B_{\mathbf{w}}(\tau')$ for every
$\mathbf{w}\in\Theta$ if and only if there exists a global unitary
$\bsU\in\rU(d_Ad_B)$ such that
\begin{eqnarray*}
\tau'=\bsU\tau \bsU^\dagger,
\quad
[\bsU,\bsX_1]=[\bsU,\bsX_2]=\mathbf{0}.
\end{eqnarray*}

A particularly important simplification occurs when both marginals are maximally mixed. If $\rho_A=\I_A/d_A$ and $\rho_B=\I_B/d_B$, then $\bsX_1$ and $\bsX_2$ are scalar multiples of the identity and commute with $\bsX_0$. Consequently
\begin{eqnarray*}
B_{\mathbf{w}}(\rho)=d^{-N_1(\mathbf{w})}_Ad^{-N_2(\mathbf{w})}_B\Tr{\rho^{N_0(\mathbf{w})}_{AB}},
\end{eqnarray*}
On this locally maximally mixed sector the complete infinite family depends only on the global spectral moments of $\rho_{AB}$. In particular, any two locally maximally mixed states with the same global spectrum agree on every ordinary marginal-word invariant. This observation underlies the constructions in Sections~\ref{sect:4} and \ref{sect:5}.

\subsection{Partial-transpose moments of qubit-qudit states}

We now specialize to states on $\cH_A\ot\cH_B=\bbC^2\ot\bbC^d$. Let $\Gamma=\t_A$ denote partial transposition with respect to a fixed orthonormal basis of the qubit subsystem $A$, and write
$$
\sigma_x=\Pa{\begin{array}{cc}
0 & 1 \\
1 & 0
\end{array}},\sigma_y=\Pa{\begin{array}{cc}
0 & -\mathrm{i} \\
\mathrm{i} & 0
\end{array}}, \sigma_z=\Pa{\begin{array}{cc}
1 & 0 \\
0 & -1
\end{array}}.
$$
The following identity converts the partial transpose into
an expression involving only the original state and one
marginal. It is the reduction criterion of Cerf, Adami,
and Gingrich and of Horodecki and Horodecki, specialized
to a qubit factor \cite{Cerf1999,Horodecki1999}.
\begin{lem}\label{lem:1}
For every state $\rho_{AB}\in\rD(\bbC^2\ot\bbC^d)$, it holds that
\begin{eqnarray}\label{eq:1}
(\sigma_y\ot\I_B)\rho^\Gamma_{AB}(\sigma_y\ot\I_B) = \I_2\ot\rho_B -
\rho_{AB}.
\end{eqnarray}
Consequently, $\rho^\Gamma_{AB}\geqslant\mathbf{0}$ if
and only if $\I_2\ot\rho_B-\rho_{AB}\geqslant\mathbf{0}$. When $d=2$, 
\begin{eqnarray}\label{eq:concurrence}
&&(\sigma_y\ot\sigma_y)\rho^*_{AB}(\sigma_y\ot\sigma_y)=(\sigma_y\ot\sigma_y)\rho^\t_{AB}(\sigma_y\ot\sigma_y) \notag\\
&&= \I_2\ot\I_2+\rho_{AB} -
\rho_A\ot\I_2-\I_2\ot\rho_B,
\end{eqnarray}
where $*$ and $\t$ denote complex conjugate and the full transpose, respectively.
\end{lem}

\begin{proof}
For every $2\times2$ complex matrix $\bsA$, one has
$\sigma_y\bsA^\t\sigma_y = \Tr{\bsA}\I_2 - \bsA$. Write $\rho_{AB}$ in block form as
$\rho_{AB}=\sum^1_{i,j=0}\out{i}{j}\ot \bsB_{ij}$, where each
$\bsB_{ij}$ is a $d\times d$ complex matrix. Partial transposition with respect to subsystem $A$ gives
$\rho^\Gamma_{AB}=\sum^1_{i,j=0}\out{j}{i}\ot \bsB_{ij}$. Applying the $2\times 2$ identity to the matrix unit $\out{i}{j}$, we obtain $\sigma_y\out{j}{i}\sigma_y=\delta_{ij}\I_2-\out{i}{j}$. It follows that
\begin{eqnarray*}
&&(\sigma_y\ot\I_B)\rho^\Gamma_{AB}(\sigma_y\ot\I_B) =
\sum^1_{i,j=0}[\sigma_y\out{j}{i}\sigma_y]\ot \bsB_{ij}\\
&&=\sum^1_{i,j=0}[\delta_{ij}\I_2-\out{i}{j}]\ot
\bsB_{ij}=\I_2\ot\rho_B-\rho_{AB},
\end{eqnarray*}
where $\rho_B=\sum^1_{i=0} \bsB_{ii}$. Since $\sigma_y\ot\I_B$ is unitary, the two operators in Eq.~\eqref{eq:1} are unitarily similar. They therefore have the same spectrum and are positive semidefinite simultaneously. Finally, Eq.~\eqref{eq:concurrence} follows
by applying the same identity on both factors.
\end{proof}
Lemma~\ref{lem:1} removes the partial transpose from all spectral-moment calculations. Set $\bsY=\rho^\Gamma_{AB}$. Because $\bsY$ is unitarily similar to $\bsX_2-\bsX_0$, 
\begin{eqnarray}\label{eq:PTBarg}
p_k(\bsY) =\Tr{\bsY^k}=\Tr{(\bsX_2-\bsX_0)^k}.
\end{eqnarray}
The operators $\bsX_0$ and $\bsX_2$ do not generally commute, so the ordinary binomial theorem cannot be applied. Expanding the noncommutative power gives, for every $k\geqslant 1$,
\begin{eqnarray*}
p_k(\bsY) = \sum_{\mathbf{w}\in\set{0,2}^k}(-1)^{N_0(\mathbf{w})}B_{\mathbf{w}},
\end{eqnarray*}
where $N_0(\mathbf{w})$ is the number of zeros in $\mathbf{w}$. Each occurrence of $0$ contributes a factor $-\bsX_0$, and each occurrence of $2$ contributes $\bsX_2$. The resulting terms may then be grouped by cyclic equivalence.

Additional simplifications follow from the definition of $\bsX_2$. For every $m\geqslant1$, $\bsX^m_2=\I_A\ot\rho^m_B$, hence
$\Tr{\bsX^m_2}=2\Tr{\rho^m_B}$. Moreover,
$$
\Tr{\bsX_0\bsX^{m-1}_2}=\Tr{\rho_{AB}(\I_2\ot\rho^{m-1}_B)}=\Tr{\rho^m_B}.
$$
Consequently, $B_{22\cdots2} = 2
B_{02\cdots2}$, and every pure $\bsX_2$-word arising in the expansion of $p_k(\bsY)$ may be replaced by a word containing $\bsX_0$.

The next lemma connects these moment expansions with the characteristic polynomial of the partial transpose.
\begin{lem}\label{lem:2}
Let $\rho_{AB}$ be a state acting on
$\bbC^2\ot\bbC^d$, and let $\bsY=\rho^\Gamma_{AB}$. Write the characteristic polynomial of $\bsY$ as
\begin{eqnarray*}
f_{\bsY}(\lambda)&=& \det(\lambda\I_{2d}-\bsY)=  \sum^{2d}_{k=0}(-1)^k e_k(\bsY)\lambda^{2d-k},
\end{eqnarray*}
where $e_k(\bsY)$ is the $k$-elementary symmetric polynomial of the eigenvalues of $\bsY$, with $e_0(\bsY)\equiv1$. Then, for every $k=0,1,\ldots,2d$, the coefficient $e_k(\bsY)$ can be expressed as a polynomial in ordinary marginal-word Bargmann invariants formed from the pair
$(\bsX_0,\bsX_2)=(\rho_{AB},\I_A\ot\rho_B)$.
\end{lem}

\begin{proof}
Let $p_j(\bsY)=\Tr{\bsY^j}$ for $j\geqslant1$. Newton’s identities relate the elementary symmetric polynomials $e_k(\bsY)$ to the power sums $p_j(\bsY)$ \cite{Macdonald1995}:
$$
ke_k(\bsY)=\sum^k_{j=1}(-1)^{j-1}e_{k-j}(\bsY)p_j(\bsY),\quad k\geqslant1,
$$
which can also be expressed more explicitly as that given in Proposition~\ref{prop:1} (see Appendix~\ref{app:A}), with $e_0(\bsY)=1$. Therefore, for each $k$, the coefficient $e_k(\bsY)$ is a polynomial in $p_1(\bsY),\ldots,p_k(\bsY)$. By Eq.~\eqref{eq:PTBarg}, every power sum $p_j(\bsY)$ is a finite linear combination of ordinary trace words formed from $\bsX_0$ and $\bsX_2$. Substituting these trace-word expansions into Newton's identities expresses $e_k(\bsY)$ as a polynomial in ordinary marginal-word Bargmann invariants.
\end{proof}

\begin{cor}
For every qubit-qudit state $\rho_{AB}\in\rD(\bbC^2\ot\bbC^d)$, the complete spectrum of the partial transpose $\rho^\Gamma_{AB}$  is determined by the ordinary marginal-word Bargmann invariants.
\end{cor}

\begin{proof}
By Lemma~\ref{lem:2},  the ordinary Bargmann invariants determine every coefficient of the characteristic polynomial of $\rho^\Gamma_{AB}$. The characteristic polynomial determines the spectrum, including algebraic multiplicities.
\end{proof}

\begin{cor}
For every qubit–qudit state, the ordinary marginal-word Bargmann invariants completely determine whether the state satisfies the PPT condition.
\end{cor}

\begin{proof}
Indeed, because $\rho^\Gamma_{AB}$ is Hermitian, $\rho^\Gamma_{AB}\geqslant\mathbf{0}$ if and only if all its eigenvalues are nonnegative. Equivalently, by Proposition~\ref{prop:1} in Appendix~\ref{app:A},
$$
e_k(\rho^\Gamma_{AB})\geqslant0,\quad k=0,1,\ldots,2d.
$$
Each of these inequalities can be written entirely in terms of ordinary marginal-word Bargmann invariants by Lemma~\ref{lem:2}.
\end{proof}
These conclusions are PPT statements, not general
separability theorems. In the $2\ot2$ and $2\ot3$ systems, PPT is equivalent to separability \cite{Peres1996,Horodecki1996}, so the ordinary invariants provide complete entanglement criteria. For $2\ot d$ system with $d\geqslant4$, PPT remains necessary but is not sufficient for separability. In those dimensions the same invariants
still determine the complete partial-transpose spectrum,
but that information alone does not decide separability.

\section{The two-qubit system}\label{sect:3}

For $(d_A,d_B)=(2,2)$, the ordinary marginal-word Bargmann invariants provide a unified description of separability, LU-orbit equivalence, and the polynomial LU-invariant ring.

\begin{thrm}\label{th:1}
Let $\rho_{AB}\in\rD(\bbC^2\ot\bbC^2)$ and write
$$
\bsX_0=\rho_{AB},\bsX_1=\rho_A\ot\I_2,\bsX_2=\I_2\ot\rho_B.
$$
The state $\rho_{AB}$ is entangled if and only if
\begin{eqnarray}\label{eq:3}
&&6 (B_{0000}-4 B_{0002}+4 B_{0022}+2 B_{0202}-2 B_{0222})\notag\\
&&+8 (B_{000}-3 B_{002}+B_{022})-3 B_{00}^2+6 B_{00}>1.
\end{eqnarray}
Explicitly, $\rho_{AB}$ is entangled if and only if
\begin{widetext}
\begin{eqnarray}\label{eq:4}
&&6\Pa{\Tr{\rho_{AB}^4}-4 \Tr{\rho_{AB}^3(\I_A\ot\rho_B)}+4 \Tr{\rho_{AB}^2(\I_A\ot\rho_B^2)}+2 \Tr{\rho_{AB}(\I_A\ot\rho_B)\rho_{AB}(\I_A\ot\rho_B)}-2\Tr{\rho_B^4}}\notag\\
&&  + 8 \Pa{\Tr{\rho_{AB}^3}-3 \Tr{\rho_{AB}^2(\I_A\ot\rho_B)}+\Tr{\rho_B^3)}} - 3 \Tr{\rho_{AB}^2}^2 + 6 \Tr{\rho_{AB}^2}>1.
\end{eqnarray}
\end{widetext}
\end{thrm}

\begin{proof}
Let $\bsY=\rho^\Gamma_{AB}$, where $\Gamma=\t_A$ denotes partial transposition with respect to the qubit subsystem $A$. By Lemma~\ref{lem:1}, $\bsY$ is unitarily similar to $\bsX_2-\bsX_0$. Hence $p_k(\bsY)=\Tr{\bsY^k}=\Tr{(\bsX_2-\bsX_0)^k}$. For a $4\times 4$ matrix with $p_1(\bsY)=\Tr{\bsY}=1$, Newton's identities give
$$
 \det(\bsY)=e_4(\bsY)=\tfrac{1}{24}\Pa{p^4_1-6p^2_1p_2+3p^2_2+8p_1p_3-6p_4}.
$$
The required power sums are
\begin{eqnarray*}
p_2(\bsY)&=&B_{00},\\
p_3(\bsY)&=&-B_{000}+3B_{002}-B_{022}, \\
p_4(\bsY)&=&B_{0000}-4B_{0002}+4B_{0022}+2B_{0202}-2B_{0222}.
\end{eqnarray*}
For example, the simplifications in $p_2$ and $p_3$ use $\Tr{\bsX^m_2}=2\Tr{\rho^m_B}=2\Tr{\bsX_0\bsX^{m-1}_2}$. Substitution into the determinant formula yields that $24\det(\rho^\Gamma_{AB}) $ is equal to
\begin{eqnarray*}
&&1 -6B_{00}+3B^2_{00}+8(-B_{000}+3B_{002}-B_{022})\\
&&-6\Pa{B_{0000}-4B_{0002}+4B_{0022}+2B_{0202}-2B_{0222}}.
\end{eqnarray*}
For two-qubit states, the determinant criterion for the partial transpose is exact:
$\rho_{AB}$ is separable if and only if $\det(\rho^\Gamma_{AB})\geqslant0$. Therefore, $\rho_{AB}$ is entangled exactly when $\det(\rho^\Gamma_{AB})<0$, which is equivalent to Eq.~\eqref{eq:3}. Expanding the Bargmann words gives Eq.~\eqref{eq:4}.
\end{proof}

Only $\bsX_0=\rho_{AB}$ and $\bsX_2=\I_2\ot\rho_B$ appear in Eqs.~\eqref{eq:3} and \eqref{eq:4}. Although ordinary marginal-word invariants are defined using the full triple $(\bsX_0,\bsX_1,\bsX_2)$, the two-qubit entanglement criterion can be written using one marginal alone. By symmetry, an equivalent criterion exists in terms of $\bsX_0=\rho_{AB}$ and $\bsX_1=\rho_A\ot\I_2$.

The same framework also solves the substantially stronger LU-orbit classification problem.

\begin{thrm}[Two-qubit LU-orbit separation \cite{Zhang2025}]\label{th:2}
Let
\begin{eqnarray}\label{eq:5}
\cW&=&\Big\{00,01,02,000,012,0000,0001,0002,0012,\notag\\
&&~~~00012,001001,002002,012001,012002,\notag\\
&&~~~0120001,0120002,010010001,020020002\Big\}.
\end{eqnarray}
Two two-qubit states $\rho_{AB}$ and $\sigma_{AB}$ are LU equivalent if and only if
$B_{\mathbf{w}}(\rho)=B_{\mathbf{w}}(\sigma)$ for all $\mathbf{w}\in\cW$. 
\end{thrm}
Thus the $18$ LU Bargmann invariants indexed by $\cW$ form a finite LU-orbit separating family.

\begin{thrm}[Two-qubit polynomial invariant ring \cite{Zhang2026}]\label{th:3}
Let $\widetilde\cW=\cW\cup\set{00120001,00120002}$, where $\cW$ is given in Eq.~\eqref{eq:5}. The polynomial LU-invariant ring for two-qubit density operators is generated by $\set{B_{\mathbf{w}}:\mathbf{w}\in\widetilde\cW}$.
\end{thrm}

Theorems~\ref{th:1}-\ref{th:3} establish the exceptional completeness of the two-qubit case. Within one trace-word framework:
(i) separability is decided by a single polynomial inequality; (ii) a finite family separates all LU orbits; and (iii) a finite extension generates the full polynomial LU-invariant ring. These three conclusions are logically distinct. Their simultaneous validity is a special low-dimensional phenomenon and does not persist for qubit-qutrit states.

\section{The qubit-qutrit system}\label{sect:4}

We now consider $(d_A,d_B)=(2,3)$.  Because the partial-transpose spectrum of every qubit–qudit state is determined by ordinary marginal-word Bargmann invariants, and because PPT is equivalent to separability in $2\ot3$ system, these invariants still decide entanglement completely. Unlike the two-qubit case, more than one characteristic-polynomial coefficient is required.

\subsection{Separability criterion}

\begin{thrm}\label{th:4}
Let $\rho_{AB}\in\rD(\bbC^2\ot\bbC^3)$ and write
$$
\bsX_0=\rho_{AB},\bsX_1=\rho_A\ot\I_3,\bsX_2=\I_2\ot\rho_B.
$$
Introduce the following ordinary marginal-word polynomials:
\begin{itemize}
\item $b_1:=1$;
\item $b_2:=B_{00}$;
\item $b_3:=3B_{002}-B_{000}-B_{022}$;
\item $b_4:=B_{0000}-4B_{0002}+4B_{0022}+2B_{0202}-2B_{0222}$;
\item
$b_5:=-B_{00000}+5B_{00002}-5B_{00022}-5B_{00202}+5B_{00222}+5B_{02022}-3B_{02222}$;
\item
$b_6:=B_{000000}-6B_{000002}+6B_{000022}+6B_{000202}+3B_{002002}-6B_{000222}-6B_{002022}-6B_{002202}-2B_{020202}+6B_{002222}+6B_{020222}+3B_{022022}-4B_{022222}$.
\end{itemize}
These quantities are precisely the partial-transpose moments: $b_k=\Tr{[\rho^\Gamma_{AB}]^k}$ for $k=1,\ldots,6$. Define $\cE_k:=k!e_k(\rho^\Gamma_{AB})$ and
\begin{itemize}
\item $\cE_3 = 1-3b_2+2b_3$;
\item $\cE_4= 1-6b_2+3b^2_2+8b_3-6b_4$;
\item $\cE_5= 1-10b_2+15b_2^2 +20b_3-20b_2b_3 -30b_4+24b_5$;
\item $\cE_6=1-15b_2+45b_2^2-15b_2^3 +40b_3-120b_2b_3+40b_3^2 -90b_4+90b_2b_4
+144b_5-120b_6$.
\end{itemize}
Then $\rho_{AB}$ is separable if and only if $\cE_k\geqslant0$ for $k=3,4,5,6$. Equivalently, $\rho_{AB}$ is entangled if and only if
\begin{eqnarray}\label{eq:6}
\min(\cE_3,\cE_4,\cE_5,\cE_6)<0.
\end{eqnarray}
\end{thrm}

\begin{proof}
Set $\bsY=\rho^\Gamma_{AB}$. By Lemma~\ref{lem:1}, $\bsY$ is unitarily similar to $\bsX_2-\bsX_0$, and therefore
$p_k(\bsY)=\Tr{\bsY^k}=\Tr{(\bsX_2-\bsX_0)^k}$. Expanding these powers as noncommutative words and grouping cyclically equivalent terms gives $p_k(\bsY)=b_k$ for $k=1,\ldots,6$. Let $e_k(\bsY)$ denote the elementary symmetric polynomials of the six eigenvalues of$\bsY$. Newton’s identities yield
$$
k! e_3(\bsY)=\cE_k,\quad k=3,4,5,6.
$$
The first two conditions require no additional inequalities. Indeed, $e_1(\bsY)=\Tr{\bsY}=1$, and $e_2(\bsY)=\frac{1-p_2(\bsY)}2=\frac{1-\Tr{\rho^2_{AB}}}2\geqslant0$, because partial transposition preserves the Hilbert-Schmidt norm and $\Tr{\rho^2_{AB}}\leqslant1$. By Proposition~\ref{prop:1}, the Hermitian matrix $\bsY$ is positive semidefinite if and only if
$e_k(\bsY)\geqslant0$ for $k=1,\ldots,6$. Consequently, $\rho^\Gamma_{AB}\geqslant\mathbf{0}$ if and only if $\cE_k\geqslant0$ for $k=3,4,5,6$. Finally, the Peres-Horodecki theorem gives $\rho_{AB}$ is separable if and only if $\rho^\Gamma_{AB}\geqslant\mathbf{0}$ for $2\ot3$ states. This proves the result.
\end{proof}

Theorem~\ref{th:4} is a complete separability criterion expressed entirely in terms of the original state and its marginal. No matrix element of the partial transpose appears on the right-hand side. In contrast with the two-qubit case, no single determinant inequality suffices: positivity of a $6\times 6$ partial transpose requires simultaneous control of several characteristic-polynomial coefficients.

\subsection{Failure of local-unitary orbit separation}

The ability to decide separability does not imply complete LU-orbit discrimination.

\begin{exam}\label{exam:2ot3}
Let $\cH_A=\bbC^2$  and $\cH_B=\bbC^3$ with computational bases $\set{\ket{0},\ket{1}}_A$ and $\set{\ket{0},\ket{1},\ket{2}}_B$, respectively. Define two traceless Hermitian operators $\bsH,\bsK$ on $\bbC^2\ot\bbC^3$ by
\begin{eqnarray*}
\bsH&=&\sigma_z\ot\Pa{\proj{0}-\proj{2}},\\
\bsK&=&\out{00}{11}+\out{11}{00}+\out{01}{12}+\out{12}{01}.
\end{eqnarray*}
For $\varepsilon\in(0,1)$, set
\begin{eqnarray}\label{eq:7}
\rho(\bsH)=\tfrac{\I_6+\varepsilon\bsH}6,\quad \rho(\bsK)=\tfrac{\I_6+\varepsilon\bsK}6.
\end{eqnarray}
These states satisfy the following properties:
(i) both states are full-rank density operators;
(ii) both states have maximally mixed marginals;
(iii) they have the same global spectrum;
(iv) every ordinary marginal-word Bargmann invariant takes the same value on the two states;
(v) the two states are not LU equivalent.
\end{exam}

To verify these claims, first observe that $\op{Spec}(\bsH)=\op{Spec}(\bsK)=\set{1,1,0,0,-1,-1}$. Hence  $\op{Spec}[\rho(\bsH)]=\op{Spec}[\rho(\bsK)]=\set{\frac{1+\varepsilon}6,\frac{1+\varepsilon}6,\frac16,\frac16,\frac{1-\varepsilon}6,\frac{1-\varepsilon}6}$. For $\varepsilon\in(0,1)$, all six eigenvalues are strictly positive, so both states are full rank and globally unitarily equivalent.

Moreover, $\ptr{A}{\bsH}=\ptr{B}{\bsH}=\ptr{A}{\bsK}=\ptr{B}{\bsK}=\mathbf{0}$. Consequently, $\rho_A(\bsH)=\rho_A(\bsK)=\frac{\I_2}2$ and $\rho_B(\bsH)=\rho_B(\bsK)=\frac{\I_3}3$. For either state, $\bsX_1=\frac{\I_6}2$ and $\bsX_2=\frac{\I_6}3$. Hence, for every word $\mathbf{w}\in\set{0,1,2}^*$,
\begin{eqnarray*}
B_{\mathbf{w}}(\rho)=2^{-N_1(\mathbf{w})}3^{-N_2(\mathbf{w})}\Tr{\rho^{N_0(\mathbf{w})}}.
\end{eqnarray*}
Since $\rho(\bsH)$ and $\rho(\bsK)$ are isospectral, $\Tr{\rho(\bsH)^k}=\Tr{\rho(\bsK)^k}$ for every $k\geqslant0$, implying $B_{\mathbf{w}}(\rho(\bsH))=B_{\mathbf{w}}(\rho(\bsK))$ for every $\mathbf{w}\in\set{0,1,2}^*$. 

It remains to show that the two states are not LU equivalent. The operator $\bsH$ is a single nonzero tensor product, so its operator Schmidt-rank is $1$. On the other hand, $\bsK$ can be written as $\bsK=\out{0}{1}\ot (\out{0}{1}+\out{1}{2})+\out{1}{0}\ot (\out{1}{0}+\out{2}{1})$. The two qubit factors are linearly independent, as are the two qutrit factors, so
the operator Schmidt-rank of $\bsK$ is $2$. If $\rho(\bsH)$ and $\rho(\bsK)$ were LU equivalent, the identity components would match and one would obtain
$$
\rho(\bsK)=(\bsU_A\ot\bsU_B)\rho(\bsH)(\bsU_A\ot\bsU_B)^\dagger,
$$
for some unitaries $\bsU_A\in\rU(2)$ and $\bsU_B\in\rU(3)$, contradicting preservation of
operator-Schmidt rank. Hence $\rho(\bsH)$ is not LU equivalent to $\rho(\bsK)$.

This example proves that equality of all ordinary marginal-word Bargmann invariants is strictly weaker
than LU equivalence in the qubit–qutrit system. The failure is not caused by truncating the family at a finite
word length: the two states agree on words of every length. A fiber of the ordinary Bargmann map may
therefore contain more than one LU orbit.

\subsection{A polynomial invariant outside the Bargmann algebra}

The same pair shows that ordinary marginal-word invariants do not generate the complete polynomial LU invariant
ring. Let
$$
\cR_{2,3}=\bbC[\Herm_1(\bbC^2\ot\bbC^3)]^{\U(2)\times\U(3)}
$$
be the ring of polynomial functions on trace-one Hermitian operators that are invariant under local-unitary conjugation, let
$$
\cB_{2,3}=\bbC[B_{\mathbf{w}}:\mathbf{w}\in\set{0,1,2}^*]
$$
be the algebra generated by all ordinary marginal-word Bargmann invariants. Since every $B_{\mathbf{w}}$ is LU-invariant, $\cB_{2,3}\subset \cR_{2,3}$. The inclusion is strict.

Choose Hilbert-Schmidt orthonormal bases $\set{\bsS_i}^3_{i=1}$ of trace-less Hermitian $2\times 2$ matrices and $\set{\bsT_j}^8_{j=1}$ of traceless Hermitian $3\times 3$ matrices. The real $3\times 8$ correlation matrix of a state $\rho$ is $\bsC(\rho)=(C_{ij}(\rho))$ whose entries are determined by $C_{ij}(\rho)=\Tr{\rho(\bsS_i\ot\bsT_j)}$. Under local unitaries, $\bsC\mapsto\bsO_A\bsC\bsO^\t_B$ for orthonormal matrices $\bsO_A\in\rO(3)$ and $\bsO_B\in\rO(8)$. The squared singular values of $\bsC$, equivalently the eigenvalues of $\bsC\bsC^\t$, are therefore LU invariants. In particular the
quartic polynomial
\begin{eqnarray}\label{eq:J2}
J_2(\rho) :=e_2\Pa{\bsC\bsC^\t}=\frac12\Br{\Inner{\bsC}{\bsC}^2 - \Inner{\bsC\bsC^\t}{\bsC\bsC^\t}}
\end{eqnarray}
is a degree-four polynomial LU invariant, independent of the chosen orthonormal operator bases.

\begin{thrm}\label{th:2to3-strict}
For the qubit-qutrit system, $\cB_{2,3}\subsetneq \cR_{2,3}$. In particular, even the infinite family of all ordinary
marginal-word Bargmann invariants does not generate the complete polynomial LU-invariant ring.
\end{thrm}

\begin{proof}
As computed in Appendix~\ref{app:B}, the correlation matrix of $\rho(\bsH)$ has a single nonzero singular value $\frac\varepsilon3$, so $J_2(\rho(\bsH))=0$. The correlation matrix of $\rho(\bsK)$ has two equal nonzero singular values $\frac{\varepsilon}{3\sqrt{2}}$, so $J_2(\rho(\bsK))=\frac{\varepsilon^4}{324}>0$. Thus $J_2$ separates the two states of Example~\ref{exam:2ot3}. Every polynomial in the ordinary invariants $B_{\mathbf{w}}$ takes the same value on the pair, so $J_2\notin\cB_{2,3}$.
\end{proof}

\begin{cor}
The complete infinite family $\set{B_{\mathbf{w}}:\mathbf{w}\in\set{0,1,2}^*}$ does not separate qubit-qutrit LU orbits.
\end{cor}

\begin{cor}
Increasing the maximum word length cannot restore completeness. The states $\rho(\bsH)$ and $\rho(\bsK)$ agree on all words of all lengths.
\end{cor}

\begin{cor}
Let $\mathfrak{B}(\rho):=(B_{\mathbf{w}}(\rho))_{\mathbf{w}\in\set{0,1,2}^*}$ denote the ordinary Bargmann map. For some qubit–qutrit states, one has
$\cO_{\mathrm{LU}}(\rho)\subsetneq  \mathfrak{B}^{-1}(\mathfrak{B}(\rho))$: a fiber of $\mathfrak{B}$ may contain more than one LU orbit. Here $\cO_{\mathrm{LU}}(\rho)$ means the LU-orbit through $\rho$.
\end{cor}

\subsection{A counterexample with nondegenerate marginals}

The obstruction is not confined to the locally maximally mixed sector. There exist states with nondegenerate,
non-maximally mixed marginals whose ordinary Bargmann invariants agree, but which are not LU equivalent.

\begin{exam}
Use the ordered product basis $\set{\ket{00},\ket{01},\ket{02},\ket{10},\ket{11},\ket{12}}$, define $\rho_{AB} = \tfrac1{100}\Pa{\sum^1_{i,j=0}\out{i}{j}\ot\bsB_{ij}}$, where
\begin{eqnarray*}
\bsB_{01}=\bsB_{10}= \Pa{\begin{array}{ccc}
  0 & 1 & 1\\
  1 & 0 & 1 \\
  1 & 1 & 0
\end{array}}
\end{eqnarray*}
and $\bsB_{00}=\diag(10,12,13),\bsB_{11}=\diag(14,16,35)$. It is easily seen that $\rho_{AB}$ is a legitimate full-rank density operator. Taking the partial trace over the second system gives
$$
\rho_A=\tfrac1{100}\diag(35,65)\text{ and }\rho_B=\tfrac1{100}\diag(24,28,48).
$$
Both marginals are nondegenerate and are not completely mixed: $\rho_{AB}\in\sC(\rho_A,\rho_B)$. Define $\bsU=\diag(\mathrm{i},1,1,1,1,1)$ in the same product basis, so that $\bsU\ket{00}=\mathrm{i}\ket{00}$, and $\bsU$ acts as the identity on the other five standard basis vectors. We have $[\bsU,\rho_A\ot\I_B]=[\bsU,\I_A\ot\rho_B]=\mathbf{0}$ because all $\bsU,\rho_A\ot\I_B,\I_A\ot\rho_B$ are diagonal in the same product basis. Notice that $\bsU$ is not a local unitary in $\rU(2)\ot\rU(3)$: its phase pattern cannot be written in the form $e^{\mathrm{i}(\alpha_j+\beta_k)}$.

The transformed state $\rho'_{AB}=\bsU\rho_{AB}\bsU^\dagger$ is
$$
\rho'_{AB}=\tfrac1{100}\sum^1_{i,j=0}\out{i}{j}\ot\bsB'_{ij},\bsB'_{01}=\Pa{\begin{array}{ccc}
  0 & \mathrm{i} & \mathrm{i}\\
  1 & 0 & 1 \\
  1 & 1 & 0
\end{array}},
$$
and $\bsB'_{10}=(\bsB'_{01})^\dagger,\bsB'_{kk}=\bsB_{kk}$ for $k=0,1$. Since $\bsU$ changes only their phases and leaves all diagonal entries unchanged, $\ptr{B}{\rho'_{AB}}=\rho_A$ and $\ptr{A}{\rho'_{AB}}=\rho_B$. Hence $\rho'_{AB}\in\sC(\rho_A,\rho_B)$.

In what follows, we show that the two states $\rho'_{AB}$ and $\rho_{AB}$ are not locally unitarily equivalent. Indeed,  suppose, toward a contradiction, that there exist local unitaries $\bsU_A\in\rU(2)$ and $\bsU_B\in\rU(3)$ such that $\rho'_{AB}=(\bsU_A\ot\bsU_B)\rho_{AB}(\bsU_A\ot\bsU_B)^\dagger$. Taking partial traces gives $\rho_X=\bsU_X\rho_X\bsU^\dagger_X$ where $X=A,B$. Because both $\rho_A$ and $\rho_B$ have nondegenerate spectra, any unitary fixing them must be diagonal in their eigenbases. Therefore,
$$
\bsU_A = \diag(e^{\mathrm{i}\alpha_0},e^{\mathrm{i}\alpha_1})\text{ and }\bsU_B=\diag(e^{\mathrm{i}\beta_0},e^{\mathrm{i}\beta_1},e^{\mathrm{i}\beta_2}).
$$
Set $\theta_{jk}:=\alpha_j+\beta_k$, where $j=0,1$ and $k=0,1,2$. Then $\Innerm{jk}{\rho'}{pq}=e^{\mathrm{i}(\theta_{jk}-\theta_{pq})}\Innerm{jk}{\rho}{pq}$.
\begin{itemize}
\item $\Innerm{00}{\rho'}{11}=e^{\mathrm{i}\frac\pi2}\Innerm{00}{\rho}{11}$ means $\theta_{00}-\theta_{11}\equiv \frac\pi2\pmod{2\pi}$;
\item Similarly, $\theta_{00}-\theta_{12}\equiv\frac\pi2\pmod{2\pi}$;
\item All the other listed coherences remain positive real numbers. Hence
\begin{eqnarray*}
\theta_{01}-\theta_{10}\equiv0\pmod{2\pi},\\
\theta_{01}-\theta_{12}\equiv0\pmod{2\pi},\\
\theta_{02}-\theta_{10}\equiv0\pmod{2\pi},\\
\theta_{02}-\theta_{11}\equiv0\pmod{2\pi}.
\end{eqnarray*}
\end{itemize}
The last four equations imply $\theta_{01}\equiv \theta_{10}\equiv \theta_{12}\equiv \theta_{02}\equiv\theta_{11}\pmod{2\pi}$. Let their common value be $c$. The first equation then gives $\theta_{00}\equiv c+\frac\pi2\pmod{2\pi}$. However, since $\theta_{jk}=\alpha_j+\beta_k$, these numbers necessarily satisfy the additive rectangle identity $\theta_{00}+\theta_{11}-\theta_{01}-\theta_{10}\equiv0\pmod{2\pi}$. Indeed,
\begin{eqnarray*}
&&\theta_{00}+\theta_{11}-\theta_{01}-\theta_{10} \\
&&=\sum^1_{k=0}(\alpha_k+\beta_k) - (\alpha_0+\beta_1)- (\alpha_1+\beta_0)\\
&&=0.
\end{eqnarray*}
But the phase conditions above instead give
\begin{eqnarray*}
&&\theta_{00}+\theta_{11}-\theta_{01}-\theta_{10} \\
&&\equiv (c+\tfrac\pi2)+c-c-c\equiv\tfrac\pi2\pmod{2\pi},
\end{eqnarray*}
which is a contradiction. Therefore, no such unitaries $\bsU_A$ and $\bsU_B$ exist, i.e., $\rho'_{AB}$ is not LU equivalent to $\rho_{AB}$.
\end{exam}
A generic qubit-qutrit density operator is a $6\times6$ Hermitian trace-one matrix, so the real dimension is $6^2-1=35$. The effective local unitary group has the same orbit dimension as $\rS\rU(2)\times \rS\rU(3)$, namely $\dim(\rS\rU(2))+\dim(\rS\rU(3))=3+8=11$. For a generic state, the stabilizer is discrete modulo the central phases, and the generic LU-orbit space has dimension $35-11=24$. This count illustrates the complexity of complete qubit-qutrit classification, but the essential obstruction is structural rather than dimensional:
on the locally maximally mixed sector the ordinary invariants collapse to global spectral moments and lose the correlation information needed to resolve LU orbits inside a fixed global unitary orbit.  Indeed, for the qubit–qutrit system, the polynomial invariant ring has an extremely large finite generating set, and its explicit determination remains an open problem to this day \cite{Gerdt2011}.

\section{Two-qudit systems with local dimension at least three}\label{sect:5}

When both local dimensions are at least three, ordinary
marginal-word invariants may fail even to determine
whether a state is entangled. The underlying mechanism
is again the locally maximally mixed reduction Eq.~\eqref{eq:genericexpression}: any
two locally maximally mixed states with the same global
spectrum agree on every ordinary Bargmann invariant.
To show that these invariants cannot decide entanglement,
it is enough to construct an isospectral locally maximally
mixed pair consisting of one separable state and
one entangled state.

\subsection{A separable-NPT pair in $\bbC^3\ot\bbC^3$}

\begin{exam}\label{exam:3}
Let $\bsF=\sum_{i,j=0}^{2}\out{ij}{ji}$ be the swap operator on $\bbC^3\ot\bbC^3$,   and
define $\rho_{AB}=\frac{ \I_9-\bsF}{6}$, $\sigma_{AB}   =
\frac13\sum_{i=0}^{2}\ket{ii}\bra{ii}$. The operator $\rho_{AB}$  is the normalized projector onto the
antisymmetric subspace. Equivalently, writing $\ket{\psi^-_{ij}}=\frac{\ket{ij}-\ket{ji}}{\sqrt{2}}$ for $i<j$, one has $\rho_{AB}=\frac13\Pa{\proj{\psi^-_{01}}+\proj{\psi^-_{12}} +\proj{\psi^-_{02}}}$. Thus $\rho_{AB}$ is a density operator of rank $3$ with spectrum $\Set{\frac13,\frac13,\frac13,0,0,0,0,0,0}$, and maximally mixed marginals $\rho_A=\rho_B=\frac{\I_3}3$.

Let $\ket{\Phi_3}=\frac{1}{\sqrt{3}}(\ket{00}+\ket{11}+\ket{22})$. The partial transpose
of the swap satisfies $\bsF^\Gamma=3\proj{\Phi_3}$, so $\rho_{AB}^\Gamma=\frac{1}{6}(\I_9-3\proj{\Phi_3})$. In particular, 
$$
\Innerm{\Phi_3}{\rho_{AB}^\Gamma}{\Phi_3} = \tfrac{1-3}{6} =
-\tfrac13<0.
$$
Thus $\rho_{AB}^\Gamma$ is not positive semidefinite, and
$\rho_{AB}$ is NPT entangled.

The state $\sigma_{AB}$ is a mixture of product states, hence separable, with the same spectrum and the same maximally
mixed marginals. The two states are therefore globally unitarily equivalent, and Eq.~\eqref{eq:genericexpression} yields $B_{\mathbf{w}}(\rho_{AB})=B_{\mathbf{w}}(\sigma_{AB})$ for every word $\mathbf{w}$. They are not LU equivalent, because local unitaries preserve separability.

An explicit global unitary implementing the equivalence may be written in the orthonormal bases
\begin{eqnarray*}
\cA&=&\Set{\ket{\psi^-_{01}},\ket{\psi^-_{12}},  \ket{\psi^-_{02}},
\ket{00},   \ket{11},   \ket{22},   \ket{\psi^+_{01}},
\ket{\psi^+_{12}},  \ket{\psi^+_{02}}},\\
\cB&=&\Set{\ket{00}, \ket{11}, \ket{22}, \ket{01}, \ket{10}, \ket{12},
\ket{21}, \ket{20}, \ket{02}},
\end{eqnarray*}
where $\ket{\psi^+_{ij}}=\frac{\ket{ij}+\ket{ji}}{\sqrt{2}}$. The linear map sending the $k$-th vector of $\cA$ to the $k$-th vector of $\cB$ is unitary and conjugates $\rho_{AB}$ to $\sigma_{AB}$.
\end{exam}

Example~\ref{exam:3} already proves that the full family of ordinary
marginal-word invariants cannot decide entanglement
in the two-qutrit system. We next extend the phenomenon
to every $n\ot n$ system with $n\geqslant3$.

\subsection{Extension to every $n\geqslant3$}

The construction uses an $n$-dimensional subspace spanned by orthonormal maximally entangled vectors and containing no nonzero product vector. Let $\omega=e^{\mathrm{i}\frac{2\pi}n}$, and define the generalized shift and phase operators by
\begin{eqnarray*}
\bsX\ket{j}=\ket{j+1},\quad
\bsZ\ket{j}=\omega^j\ket{j},
\end{eqnarray*}
with indices modulo $n$.
\begin{lem}\label{lem:no-rank-one}
For every $n\geqslant 3$, there exist unitary matrices $\bsU_0,\ldots,\bsU_{n-1}\in \rU(n)$ such that 
\begin{itemize}
\item[(i)] they are pairwise Hilbert--Schmidt orthogonal, $\Inner{\bsU_i}{\bsU_j}=n\delta_{ij}$; and
\item[(ii)] every nonzero matrix in their linear span has rank at least two.
\end{itemize}
\end{lem}

\begin{proof}
Set $\bsU_k=\bsX^k \bsZ^{a_k}$ with $a_{n-1}=1$ and $a_k=0$ for $0\leqslant k\leqslant n-2$. Hilbert--Schmidt orthogonality follows from the Weyl relations. Let $\bsM=\sum_{k=0}^{n-1}c_k \bsU_k$. Its matrix elements satisfy
\begin{eqnarray}\label{eq:Mij}
\Innerm{i}{\bsM}{j}=c_{i-j} \omega^{a_{i-j}j},
\end{eqnarray}
with indices modulo $n$. Suppose, for contradiction, that $\bsM=\out{\bsx}{\bsy}$ is a nonzero rank-one matrix. Choose $k$ with $c_k\neq 0$. Eq.~\eqref{eq:Mij} implies that every entry on the $k$-th cyclic diagonal is nonzero,
\begin{eqnarray*}
m_{j+k,j}=c_k\omega^{a_k j}\neq 0,\quad 0\leqslant j\leqslant n-1.
\end{eqnarray*}
Hence every coordinate of $\bsx$ and of $\bsy$ is nonzero, so every cyclic diagonal is nonzero and $c_i\neq 0$ for all $i$. Since $\bsM$ has rank one, every $2\times 2$ minor vanishes. Using rows $j+k, j+k+1$ and columns $j, j+1$ gives
\begin{eqnarray*}
c_k^2\omega^{a_k(2j+1)}=c_{k-1}c_{k+1}\omega^{a_{k-1}(j+1)+a_{k+1}j}.
\end{eqnarray*}
Replacing $j$ by $j+1$ and dividing the two identities yields the cocycle relation
\begin{eqnarray*}
2a_k\equiv a_{k-1}+a_{k+1}\pmod n.
\end{eqnarray*}
For $k=0$ this becomes $2a_0\equiv a_{n-1}+a_1\pmod n$. By construction the right-hand side equals$1$ and the left-hand side equals$0$, which is impossible for $n\geqslant 3$. Thus $\rank(\bsM)\geqslant 2$.
\end{proof}

Let $\ket{\Phi_n}=\frac1{\sqrt{n}}\sum_{j=0}^{n-1}\ket{jj}$ be the standard maximally entangled vector. For any $\bsM\in M_n(\bbC)$ the vector $(\I\ot\bsM)\ket{\Phi_n}$ has Schmidt rank equal to $\rank(\bsM)$, and is a nonzero product vector if and only if $\bsM$ has rank one. Define
\begin{eqnarray}\label{eq:psik}
\ket{\psi_k}=(\I\ot \bsU_k)\ket{\Phi_n},\quad 0\leqslant k\leqslant n-1,
\end{eqnarray}
with $\bsU_k$ as in Lemma~\ref{lem:no-rank-one}. Each$\ket{\psi_k}$ is maximally entangled, and
\begin{eqnarray*}
\Inner{\psi_i}{\psi_j}=\tfrac1n\Tr{\bsU_i^\dagger \bsU_j}=\delta_{ij},
\end{eqnarray*}
so these vectors form an orthonormal family. Let $\cS=\op{Span}\set{\ket{\psi_k}:k=0,1,\ldots,n-1}$. If a nonzero product vector $\ket{\bsa},\ket{\bsb}$ belonged to $\cS$, then for some coefficients not all zero one would have
\begin{eqnarray*}
\ket{\bsa}\ot\ket{\bsb}
=\Pa{\I\ot\sum_{k=0}^{n-1}c_k \bsU_k}\ket{\Phi_n},
\end{eqnarray*}
forcing $\sum_{k=0}^{n-1}c_k \bsU_k$ to have rank one and contradicting Lemma~\ref{lem:no-rank-one}. Hence $\cS$ contains no nonzero product vector.

We recall a low-rank PPT theorem of Horodecki, Lewenstein, Vidal, and Cirac~\cite{Horodecki2000}.
\begin{thrm}[\cite{Horodecki2000}]\label{thm:hlvc}
Let $\rho_{AB}\in\rD(\bbC^m\ot\bbC^n)$ with $m\leqslant n$. If $\rho_{AB}^{\Gamma}\geqslant \mathbf{0}$ and $\rank(\rho_{AB})=n$, then $\rho_{AB}$ is separable.
\end{thrm}

\begin{thrm}\label{thm:nn}
For every $n\geqslant3$, the states
\begin{eqnarray}\label{eq:nn}
\rho_{AB}=\frac1n\sum_{k=0}^{n-1}\proj{\psi_k},\quad
\sigma_{AB}=\frac1n\sum_{k=0}^{n-1}\proj{kk}
\end{eqnarray}
have the same global spectrum and maximally mixed marginals, where $\ket{\psi_k}$ is taken from Eq.~\eqref{eq:psik}. The state $\sigma_{AB}$ is separable, whereas $\rho_{AB}$ is NPT entangled. Consequently $B_{\mathbf{w}}(\rho_{AB})=B_{\mathbf{w}}(\sigma_{AB})$ for every ordinary marginal word $\mathbf{w}$, although the states are not LU equivalent.
\end{thrm}

\begin{proof}
The vectors $\ket{\psi_k}$ are orthonormal, as are the product vectors $\ket{kk}$, so both states have spectrum consisting of $n$ eigenvalues $\frac1n$ and $n^2-n$ zeros. The state $\sigma_{AB}$ is a mixture of product projectors and is therefore separable; its marginals equal $\frac{\I_n}n$. Each $\ket{\psi_k}$ is maximally entangled, so averaging over $k$ likewise yields $\rho_A=\rho_B=\frac{\I_n}n$. Eq.~\eqref{eq:genericexpression} therefore implies that all ordinary marginal-word invariants agree.

The range of $\rho_{AB}$ is exactly $\cS$, which contains no nonzero product vector. If $\rho_{AB}$ were separable, every product vector appearing with nonzero weight in a pure-state decomposition would lie in $\cS$, which is impossible. Thus $\rho_{AB}$ is entangled. Its rank is $n$. Theorem~\ref{thm:hlvc} with $m=n$ implies that a PPT state of rank $n$ on $\bbC^n\ot\bbC^n$ is separable. An entangled state of rank $n$ therefore cannot be PPT, and $\rho_{AB}$ is NPT. Local unitaries preserve separability, so the two states are not LU equivalent.
\end{proof}

The use of Theorem~\ref{thm:hlvc} is needed only to upgrade entanglement to NPT entanglement. Absence of product vectors in the image of $\rho_{AB}$ already proves that $\rho_{AB}$ is entangled.

\begin{exam}\label{ex:n3-num}
For $n=3$, take $\omega=e^{\mathrm{i}\frac{2\pi}3}$ and $\bsU_0=\I_3, \bsU_1=\bsX, \bsU_2=\bsX^2 \bsZ$. The corresponding maximally entangled vectors are
\begin{eqnarray*}
\ket{\psi_0}&=&\tfrac{1}{\sqrt{3}}\Pa{\ket{00}+\ket{11}+\ket{22}},\\
\ket{\psi_1}&=&\tfrac{1}{\sqrt{3}}\Pa{\ket{01}+\ket{12}+\ket{20}},\\
\ket{\psi_2}&=&\tfrac{1}{\sqrt{3}}\Pa{\ket{02}+\omega\ket{10} + \omega^2\ket{21}}.
\end{eqnarray*}
The mixture $\rho_{AB}=\frac13\sum_{k=0}^2\proj{\psi_k}$ has maximally mixed marginals. Numerical diagonalization of $\rho^\Gamma_{AB}$ yields the spectrum
\begin{equation*}
\begin{split}
&\bigl\{0.281343,\;0.281343,\;0.281343,\\
&\quad 0.149700,\;0.149700,\;0.149700,\\
&\quad -0.097710,\;-0.097710,\;-0.097710\bigr\},
\end{split}
\end{equation*}
confirming three negative eigenvalues and hence NPT entanglement, in agreement with Theorem~\ref{thm:nn}.
\end{exam}
Note that in the case $m=n\geqslant 3$, more general and abstract constructions are available; see \cite{Ma2026}.

\section{Supplement to Theorem 8: Transcendence-degree argument on the locally maximally mixed slice}\label{sect:6}

Let $\cR_{m,n}$ denote the polynomial ring of trace-one Hermitian operators on $\bbC^m\ot\bbC^n$ that are invariant under $\rU(m)\times \rU(n)$, i.e.,
$$
\cR_{m,n}=\bbC[\Herm_1(\bbC^m\ot\bbC^n)]^{\U(m)\times\U(n)}
$$
and let $\cB_{m,n}$ be the subalgebra generated
by all ordinary marginal-word invariants, $\cB_{m,n}=\bbC[B_{\mathbf{w}}:\mathbf{w}\in\set{0,1,2}^*]$.

Recall that for a finitely generated domain $\cA$ over $\bbC$ 
with field of fractions $\op{Frac}(\cA)$, the \emph{transcendence degree}  (formal definition is put in Appendix~\ref{app:C})
$\op{trdeg}_\bbC(\cA)$ is the maximal number of algebraically 
independent elements in $\op{Frac}(\cA)$ over $\bbC$. 
Geometrically, for an invariant ring $\bbC[V]^G$ ($V$ a vector space) under a reductive group action, 
$\op{trdeg}_{\bbC}(\bbC[V]^G)$ coincides with the dimension of the 
generic orbit space. If $\cA_1 \subset \cA_2$ are two subalgebras 
and $\op{trdeg}_{\bbC}(\cA_1) < \op{trdeg}_{\bbC}(\cA_2)$, 
then $\cA_1$ is necessarily a proper subalgebra of $\cA_2$.

\begin{thrm}\label{th:8}
For integers $m,n\geqslant2$, 
\begin{eqnarray}\label{eq:properinclusion}
\cB_{m,n}=\cR_{m,n} \text{ if and only if } (m,n)=(2,2),
\end{eqnarray}
and $\cB_{m,n}\subsetneq\cR_{m,n}$ strictly whenever $\max(m,n)\geqslant3$.
\end{thrm}

\begin{proof}
The sufficiency for $(m,n)=(2,2)$ is established by Theorem~\ref{th:3}. For any 
other pair $(m,n)\neq(2,2)$ with $m,n\geqslant2$, it suffices to demonstrate that the field of fractions of $\cB_{m,n}$ has a strictly smaller transcendence degree than that of $\cR_{m,n}$.

We restrict both invariant algebras to the affine slice of locally maximally mixed states, $V_{m,n}\equiv\sC(\tfrac{\I_m}m,\tfrac{\I_n}n)$, that is,
$$
V_{m,n}:=\Set{\rho\in\rD(\bbC^m\ot\bbC^n): \ptr{B}{\rho}=\tfrac{\I_m}m,\ptr{A}{\rho}=\tfrac{\I_n}n}.
$$
where $D:=mn$. The underlying translation space is
$\widetilde V_{m,n}:=V_{m,n}-\tfrac{\I_D}D=\Herm_0(\bbC^m)\ot_{\bbR}\Herm_0(\bbC^n)$, having real dimension
$d=\dim_{\bbR}\widetilde V_{m,n}=(m^2-1)(n^2-1)$. On this slice, the lifted marginals collapse to scalar multiples of the identity: $\bsX_1=\frac1m\I_D$ and $\bsX_2=\frac1n\I_D$ with $D=mn$. Consequently, every ordinary marginal-word Bargmann invariant reduces to a scalar multiple of a global spectral moment,
$$
B_{\mathbf{w}}(\rho) = m^{-N_1(w)} n^{-N_2(w)} \Tr{\rho^{N_0(\mathbf{w})}}.
$$
Since $\Tr{\rho}=1$, the restricted algebra $\B_{m,n}|_{V_{m,n}}$ is generated solely by the $D-1$ global trace powers $\set{\Tr{\rho^k}}^D_{k=2}$, whence 
$$
\op{trdeg}\Pa{\cB_{m,n}|_{V_{m,n}}}\leqslant D-1=mn-1.
$$
Conversely, the effective local unitary group acting on $V_{m,n}$ is $G=\op{PU}(m)\times\op{PU}(n)$, whose Lie algebra dimension is 
$g=(m^2-1)+(n^2-1)=m^2+n^2-2$. Because each local-unitary orbit has dimension at most $g$, the generic LU orbit space on $\widetilde V_{m,n}$ has dimension at least $d-g$.  Since the orthogonal projection onto $\widetilde V_{m,n}$ is $G$-equivariant, any  local-unitary invariant on the slice extends to a polynomial invariant in $\cR_{m,n}$ (see Appendix~\ref{app:C}). The dimensional difference satisfies
\begin{eqnarray*}
Q_{m,n}&:=&(d-g)-(mn-1)\\
&=&(m^2-2)(n^2-2)-mn.
\end{eqnarray*}
Direct inspection shows $Q_{2,2}=0$, whereas $Q_{m,n}>0$ for all $(m,n)\neq(2,2)$ (e.g., $Q_{2,3}=8$ and $Q_{3,3}=40$). It follows that
\begin{eqnarray*}
\op{trdeg}\Pa{\cR_{m,n}|_{V_{m,n}}}>\op{trdeg}\Pa{\cB_{m,n}|_{V_{m,n}}},
\end{eqnarray*}
which guarantees that $\cB_{m,n}$ is a proper subalgebra of $\cR_{m,n}$.  This completes the proof.
\end{proof}

\section{Conclusions}\label{sect:7}

Ordinary marginal-word Bargmann invariants associated with the canonical triple $(\bsX_0,\bsX_1,\bsX_2)$ are natural polynomial LU invariants, but they do not provide a dimension-independent description of bipartite states. Their exact scope is as follows.

\begin{enumerate}
\item For every qubit--qudit state they determine the complete partial-transpose spectrum and therefore the PPT property.
\item In $2\ot 2$ and $2\ot 3$ systems, they decide separability, because PPT is equivalent to separability in those dimensions.
\item In $2\ot 2$ system, a finite subfamily separates LU orbits, and a finite extension generates the polynomial invariant ring.
\item In $2\ot 3$ system, all ordinary invariants can agree on distinct LU orbits, and a quartic correlation invariant lies outside their algebra. The same phenomenon occurs for certain states with nondegenerate, non-maximally mixed marginals.
\item For every $n\geqslant 3$, all ordinary invariants can agree on a separable state and an NPT-entangled state with the same global spectrum and maximally mixed marginals.
\item The algebra generated by the ordinary invariants coincides with the full polynomial LU-invariant ring if and only if both subsystems are qubits.
\end{enumerate}

The decisive loss occurs on the locally maximally mixed sector. There the marginal operators carry no nontrivial eigenspace information, and the trace words collapse to global spectral moments. The missing data are not additional moments of the same operators, but correlation invariants that retain the placement of global eigenspaces relative to the tensor-product decomposition. Operator-Schmidt coefficients, correlation-matrix contractions, and higher-order tensor invariants are natural candidates.

Several problems remain open. A finite separating family for generic qubit--qutrit mixed states is not identified here; the generic orbit space has real dimension $24$, but a complete separating set must also resolve nongeneric strata. It would be useful to characterize the subsets of state space on which ordinary marginal-word invariants remain complete, for instance under nondegeneracy assumptions on the marginal spectra. A further problem is to determine which additional invariants admit efficient estimation through multivariate trace protocols or collective measurements~\cite{Oszmaniec2024,Quek2024}. The main structural lesson is that increasing the word length within the same canonical triple cannot repair an information loss caused by the collapse of the marginals to scalar operators.

\begin{acknowledgments}
The authors designed the study, carried out the mathematical analysis, and prepared the manuscript. Generative-AI tools were used for language refinement. All mathematical statements were independently reviewed and verified by the authors, who take full responsibility for the content.
\end{acknowledgments}

\section*{Data availability}
No data were created or analyzed in this study.

\appendix

\section{Positivity from elementary symmetric polynomials}\label{app:A}

This is a known result, recorded here for convenience.

\begin{prop}\label{prop:1}
For a $d\times d$ Hermitian matrix $\bsH$, it holds that
\begin{eqnarray}
\bsH\geqslant \mathbf{0}\Longleftrightarrow e_k(\bsH)\geqslant0\text{ for all
}k=1,\ldots,d.
\end{eqnarray}
\end{prop}

\begin{proof}
If all eigenvalues are nonnegative, every elementary symmetric
polynomial is a sum of products of nonnegative numbers, so
$e_k(\bsH)\geqslant0$. Conversely, suppose that every
$e_k(\bsH)\geqslant0$. Consider
$g(t)=\det(t\I_d+\bsH)=\prod^d_{j=1}(t+\lambda_j)$, which is equal
to
$$
g(t)=t^d +e_1(\bsH)t^{d-1}+\cdots+e_d(\bsH).
$$
For every $t>0$, all terms on the right are nonnegative and the
leading term is strictly positive. Hence $g(t)>0$ for $t>0$. If some
eigenvalue $\lambda_j$ were negative, then $t=-\lambda_j>0$ would
satisfy $g(-\lambda_j)=0$, which is impossible. Therefore all
eigenvalues are nonnegative.
\end{proof}
The power sums $p_k(\bsH)=\Tr{\bsH^k}$ and the elementary
symmetric polynomials are related by Newton’s identities $ke_k(\bsY)=\sum^k_{j=1}(-1)^{j-1}e_{k-j}(\bsY)p_j(\bsY)$ or equivalently by the determinantal formula \cite{Macdonald1995}:
\begin{eqnarray*}
e_k(\bsH)=\frac1{k!}\Abs{\begin{array}{ccccc}
                           p_1(\bsH) & 1 & 0 & \cdots & 0 \\
                           p_2(\bsH) & p_1(\bsH) & 2 & \cdots & 0 \\
                           \vdots & \vdots & \vdots & \ddots & \vdots \\
                           p_{k-1}(\bsH) & p_{k-2}(\bsH) & p_{k-3}(\bsH) & \cdots & k-1 \\
                           p_k(\bsH) & p_{k-1}(\bsH) & p_{k-2}(\bsH) & \cdots & p_1(\bsH)
                         \end{array}
}.
\end{eqnarray*}
This is the only positivity input used in Sections~\ref{sect:2} and \ref{sect:4}.

\section{The quartic correlation invariant}\label{app:B}

Let $\Herm_0(\bbC^d)$ be the real vector space of traceless Hermitian $d\times d$ operators, and
\begin{eqnarray*}
&&\set{\bsS_i:i=1,2,3}\subset \Herm_0(\bbC^2),\\
&&\set{\bsT_j:j=1,\ldots,8}\subset\Herm_0(\bbC^3),
\end{eqnarray*}
be Hilbert–Schmidt orthonormal bases of the traceless Hermitian operator spaces, and let $\bsC\equiv\bsC(\rho)$ be the correlation
matrix whose entries are given by $\Tr{\rho(\bsS_i\ot\bsT_j)}$. Under local conjugation the adjoint representations act orthogonally, so
$$
\bsC\mapsto\bsO_A\bsC\bsO^\t_B, \bsC\bsC^\t\mapsto\bsO_A(\bsC\bsC^\t)\bsO^\t_A.
$$
The eigenvalues of $\bsC\bsC^\t$ are therefore polynomial LU invariants, as are all of their elementary symmetric polynomials.
The quantity $J_2$ in Eq.~\eqref{eq:J2} is the second of these.

For $\bsH=\sigma_z\ot(\proj{0}-\proj{2})$, one has $\norm{\sigma}_2=\sqrt{2}$ and $\norm{\diag(1,0,-1)}_2=\sqrt{2}$. The correlation matrix of 
$\rho(\bsH)$ therefore has a single nonzero singular value,
\begin{eqnarray*}
s_1(C(\rho(\bsH)))=\tfrac{\varepsilon}3,\quad s_2=s_3=0.
\end{eqnarray*}
Hence $J_2(\rho(\bsH))=\sum_{1\leqslant i<j\leqslant 3}s_i^2 s_j^2=0$.

For $\bsK$,  the operator-Schmidt decomposition into two independent Hermitian factors produces two equal nonzero singular values
\begin{equation}
s_1(C(\rho(\bsK)))=s_2(C(\rho(\bsK)))=\tfrac{\varepsilon}{3\sqrt{2}},\quad
s_3=0.
\end{equation}
Therefore
\begin{equation}
J_2(\rho(\bsK))=s_1^2 s_2^2=\tfrac{\varepsilon^4}{324}>0.
\end{equation}
Since $B_{\mathbf{w}}(\rho(\bsH))=B_{\mathbf{w}}(\rho(\bsK))$ for every word $\mathbf{w}$, no polynomial in the ordinary Bargmann invariants can reproduce $J_2$. This proves Theorem~\ref{th:2to3-strict} directly.

The obstruction can be understood without reference to this particular quartic. When both marginals are maximally mixed, all ordinary Bargmann invariants collapse to spectral moments and retain only the global spectrum. The full LU-invariant ring, by contrast, contains invariants sensitive to the bipartite correlation structure: operator-Schmidt coefficients, polynomial functions of the correlation matrix, coefficients of the characteristic polynomial of $\bsC\bsC^\t$, and higher-order contractions of local correlation tensors. For the states of Example~\ref{exam:2ot3}, the traceless correlation operators have different operator-Schmidt ranks, $1$ for $\bsH$ and $2$ for $\bsK$. The invariant $J_2$ is a polynomial detector of precisely this difference: $J_2(\rho)=0$ whenever the correlation matrix has rank at most one, while $J_2(\rho)>0$ for the rank-two correlation matrix of $\rho(\bsK)$.

\section{Transcendence-degree analysis on the locally maximally mixed slice}\label{app:C}

In this Appendix, we provide the complete algebraic and geometric justification for Theorem~\ref{th:8}. The proof relies on comparing the transcendence degrees of the invariant algebras restricted to the locally maximally mixed slice, establishing that ordinary Bargmann invariants fail to separate local-unitary orbits for any bipartite system beyond two-qubit states.

Before proceeding, let recall the formal definition \cite{Mumford1994}.
\begin{definition}[Transcendence degree]
Let $\mathbbm{k}$ be a field and $\cA$ an integral domain containing $\mathbbm{k}$, with field of fractions $\cK = \op{Frac}(\cA)$. A finite subset $\Set{f_1, \ldots, f_r} \subset \cK$ is said to be \emph{algebraically independent} over $\mathbbm{k}$ if there is no non-zero polynomial $P \in \mathbbm{k}[x_1, \dots, x_r]$ such that $P(f_1, \ldots, f_r) = 0$. The \emph{transcendence degree} of $\cA$ (or of $\cK$) over  $\mathbbm{k}$, denoted by $\op{trdeg}_{\mathbbm{k}}(\cA)$ (or $\op{trdeg}_{\mathbbm{k}} \cK$), is the maximal cardinality of an algebraically independent subset of $\cK$ over $\mathbbm{k}$. Geometrically, if $\cA = \bbC[V]^G$ is the ring of polynomial invariants of an affine variety $V$ under an algebraic group $G$, $\op{trdeg}_{\bbC}(\cA)$ equals the geometric dimension of the categorical quotient variety $V  /\!\!/ G$ (i.e., the dimension of the generic orbit space).
\end{definition}

\subsection{The locally maximally mixed slice and equivariant extension}

Let $\cH := \bbC^m \ot\bbC^n$ with total dimension $D = mn$. The real vector space $\Herm_1(\cH)$ of Hermitian operators of trace one admits the canonical orthogonal decomposition under the Hilbert–Schmidt inner product:
\begin{eqnarray}\label{eq:decomp}
\Herm_1(\cH) &=& \tfrac{\I_D}D \oplus \Br{\Herm_0(\bbC^m)\ot\I_n}\notag\\
&&\oplus\Br{\I_m\ot\Herm_0(\bbC^n)}\oplus \widetilde V_{m,n},
\end{eqnarray}
where
\begin{eqnarray*}
\widetilde{V}_{m,n} = \Herm_0(\bbC^m) \ot_{\bbR} \Herm_0(\bbC^n).
\end{eqnarray*}
Here, $\Herm_0(\bbC^k)$ denotes the $(k^2-1)$-dimensional space of traceless Hermitian $k\times k$ matrices. The affine subspace of density operators with maximally mixed marginals is defined by
\begin{eqnarray*}
V_{m,n} &=& \Set{\rho\in\rD(\bbC^m\ot\bbC^n): \ptr{B}{\rho}=\tfrac{\I_m}m,\ptr{A}{\rho}=\tfrac{\I_n}n}\\
&=& (\tfrac{\I_D}D+ \widetilde V_{m,n})\cap \rD(\cH).
\end{eqnarray*}
The real dimension of $\widetilde{V}_{m,n}$ is
$$
d := \dim_{\bbR} \widetilde{V}_{m,n} = (m^2-1)(n^2-1).
$$
Let $G=\op{PU}(m)\times \op{PU}(n)$  be the effective local unitary group. The decomposition in Eq.~\eqref{eq:decomp} is manifestly $G$-equivariant. Let $\pi:\Herm_1(\cH)\to V_{m,n}$ be the orthogonal affine projection that annihilates the local traceless components $\Herm_0(\bbC^m)\ot\I_n$ and $\I_m\ot\Herm_0(\bbC^n)$. Since $\pi(\bsU\rho\bsU^\dagger)=\bsU\pi(\rho)\bsU^\dagger$ for all $\bsU\in\rU(m)\times\rU(n)$, any $G$-invariant polynomial $f$ on $V_{m,n}$ extends to a full local-unitary invariant polynomial $F=f\circ\pi\in\cR_{m,n}$. Consequently, the restriction homomorphism
$$
\op{Res}_V:\cR_{m,n}\to \bbC[V_{m,n}]^G
$$
is surjective. Therefore, the transcendence degree of $\cR_{m,n}$ restricted to $V_{m,n}$ coincides with the transcendence degree of the invariant ring $\bbC[V_{m,n}]^G$.

\subsection{Upper bound for ordinary Bargmann invariants}

For any $\rho\in V_{m,n}$, the marginal operators are proportional to the identity:
$\bsX_0:=\rho,  \bsX_1:=\rho_A\ot\I_n = \tfrac{1}{m}\I_D, 
\bsX_2:=\I_m\ot \rho_B = \tfrac{1}{n}\I_D$. Because $\bsX_1$ and $\bsX_2$ are central in the full matrix algebra, they commute with $\bsX_0=\rho$. For any finite word $\mathbf{w}\in \set{0,1,2}^*$, let $N_k(\mathbf{w})$ denote the number of occurrences of the symbol $k$ in $\mathbf{w}$. Then the ordinary marginal-word Bargmann invariant is
\begin{eqnarray*}
B_{\mathbf{w}}(\rho) = m^{-N_1(\mathbf{w})} n^{-N_2(\mathbf{w})} \Tr{\rho^{N_0(\mathbf{w})}}.
\end{eqnarray*}

\begin{prop}
The restricted algebra $\cB_{m,n}|_{V_{m,n}}$ satisfies
\begin{eqnarray*}
\op{trdeg}\Pa{\cB_{m,n}|_{V_{m,n}}}\leqslant D-1=mn-1.
\end{eqnarray*}
\end{prop}
\begin{proof}
All generators of $\cB_{m,n}|_{V_{m,n}}$ are scalar multiples of the global spectral moments $p_k(\rho)=\Tr{\rho^k}$. By the Cayley-Hamilton theorem and Newton's identities, the characteristic polynomial of any $D\times D$ matrix is determined by its first $D$ moments. Since $p_1(\rho)=\Tr{\rho}=1$ is constant on $V_{m,n}$, at most $D-1$ moments $\set{p_k(\rho)}^D_{k=2}$ are algebraically independent. Thus, the field of fractions of  $\cB_{m,n}|_{V_{m,n}}$ has transcendence degree at most $D-1$.
\end{proof}

\subsection{Lower bound for the full invariant ring}

The group $G := \op{PU}(m)\times \op{PU}(n)$ is a compact Lie group of real dimension
\begin{eqnarray*}
g := \dim_{\bbR} G = (m^2-1)+(n^2-1) = m^2+n^2-2.
\end{eqnarray*}

\begin{prop}\label{lem:orbitdim}
 The transcendence degree of the full invariant ring restricted to $V_{m,n}$ satisfies
\begin{eqnarray*}
\op{trdeg}\Pa{\cR_{m,n}|_{V_{m,n}}}\geqslant d-g.
\end{eqnarray*}
\end{prop}

\begin{proof}
The translation space $\widetilde V_{m,n}$ is a smooth Euclidean manifold of dimension $d$. Under the smooth adjoint action of the compact Lie group $G$, the orbit $\cO_\rho=\Set{\bsU\rho\bsU^\dagger:\bsU\in G}$ through any point $\rho\in V_{m,n}$, has dimension
$$
\dim_{\bbR}(\cO_\rho)\leqslant \dim_{\bbR}(G)=g.
$$
By classical geometric invariant theory (or the slice theorem for compact transformation groups), the field of invariant rational functions $\bbR[V_{m,n}]^G$ separates generic $G$-orbits. Its transcendence degree over $\bbR$ equals the dimension of the generic orbit space (topological quotient):
\begin{eqnarray*}
\op{trdeg}_{\bbR}\bbR[V_{m,n}]^G&=&\dim_{\bbR}(V_{m,n})-\max_{\rho\in V_{m,n}}\dim_{\bbR}(\cO_\rho)\\
&\geqslant& d-g.
\end{eqnarray*}
Upon complexification, 
$$
\op{trdeg}_{\bbC}\bbC[V_{m,n}]^G=\op{trdeg}_{\bbR}\bbR[V_{m,n}]^G\geqslant d-g.
$$
By the surjectivity established previously, the claim follows.
\end{proof}

\subsection{Gap comparison and proof of proper inclusion}

\begin{prop}
Define $Q_{m,n} :=(d-g) - (mn-1)$. For all integers $m,n\geqslant2$,
\begin{eqnarray*}
Q_{m,n} = (m^2-2)(n^2-2) - mn.
\end{eqnarray*}
Moreover, $Q_{m,n}=0$ if and only if $(m,n)=(2,2)$, and $Q_{m,n}>0$ whenever $\max(m,n)\geqslant 3$.
\end{prop}
\begin{proof}
Expanding the dimensions yields
\begin{eqnarray*}
Q_{m,n}&=& (m^2-1)(n^2-1) - (m^2+n^2-2) - mn + 1 \\
&&= m^2n^2 - 2m^2 - 2n^2 - mn + 4 \\
&&= (m^2-2)(n^2-2) - mn.
\end{eqnarray*}
If $m=2$, then $Q_{2,n}=2(n-2)(n+1)$, which equals $0$ for $n=2$ and is strictly positive for all $n\geqslant3$. By symmetry, $Q_{m,2}>0$ for $m\geqslant3$. For $m,n\geqslant3$, we have $m^2-2>m$ and $n^2-2>n$, so $(m^2-2)(n^2-2)>mn$, ensuring $Q_{m,n}>0$.
\end{proof}
Combining them above: for any $(m,n)\neq(2,2)$,
\begin{eqnarray*}
\op{trdeg}\Pa{\cR_{m,n}|_{V_{m,n}}} - \op{trdeg}\Pa{\cB{m,n}|_{V_{m,n}}} \geqslant Q_{m,n} > 0.
\end{eqnarray*}

\subsection{Geometric consequence: isospectral fibers}

The surplus dimension $Q_{m,n}>0$ directly implies that the fiber of the ordinary Bargmann map contains continuous families of inequivalent local-unitary orbits.

Let $s:V_{m,n}\to \bbR^{D-1}$ be the spectral map
$$
\rho\mapsto (p_2(\rho),\ldots,p_D(\rho)).
$$
By Sard’s theorem, a generic regular fiber $\cF_\lambda=s^{-1}(\lambda)$ is a smooth submanifold of $V_{m,n}$ with dimension
\begin{eqnarray*}
\dim_{\bbR}(\cF_\lambda)\geqslant d-(D-1).
\end{eqnarray*}
Since any local-unitary orbit $\cO$ inside $\cF_\lambda$ satisfies $\dim_{\bbR}(\cO)\leqslant g$, the quotient of the isospectral fiber by local-unitary equivalence satisfies
\begin{eqnarray*}
\dim_{\bbR}\Pa{\cF_\lambda/G}\geqslant [d-(D-1)]-g=Q_{m,n}.
\end{eqnarray*}
Whenever $Q_{m,n}>0$, each generic isospectral slice in $V_{m,n}$ contains a $Q_{m,n}$-dimensional continuum of mutually LU-inequivalent states. Because every state in $\cF_\lambda$ shares the exact same global spectrum and maximally mixed marginals, all ordinary marginal-word Bargmann invariants take identical values across the entire family.

Thus, the failure of $\cB_{m,n}$ is geometric: the invariants $B_{\mathbf{w}}$ record only the eigenvalue spectrum of $\rho$ but entirely discard the relative orientation of its eigenspaces with respect to the tensor-product factorisation $\cH_A\ot\cH_B$.




\begin{thebibliography}{99}

\bibitem{Grassl1998}
M. Grassl, M. R\"{o}tteler, T. Beth, Computing local invariants of quantum-bit systems, \pra \href{https://doi.org/10.1103/PhysRevA.58.1833}{{\bf58}, 1833 (1998).}

\bibitem{Makhlin2002}
Y. Makhlin, Nonlocal properties of two-qubit gates and mixed
states, and the optimization of quantum computations, Quantum
Inf. Process. \href{https://doi.org/10.1023/A:1022144002391}{{\bf1}, 243 (2002).}

\bibitem{Kraus2009}
B. Kraus, Local unitary equivalence of multipartite pure states, \prl \href{https://doi.org/10.1103/PhysRevLett.104.020504}{{\bf104}, 020504 (2009).}

\bibitem{Zhou2012}
C. Zhou, T. Zhang, S-M. Fei, N. Jing, and X. Li-Jost, Local unitary equivalence of arbitrary dimensional bipartite mixed quantum states, \pra~\href{https://doi.org/10.1103/10.1103/PhysRevA.86.010303}{{\bf86}, 010303(2012).}

\bibitem{Gour2013}
G. Gour and N.R. Wallach, Classification of multipartite entanglement of all finite dimensionality, \prl \href{https://doi.org/10.1103/PhysRevLett.111.060502}{{\bf111}, 060502 (2013).}

\bibitem{Turner2017}
J. Turner, and J. Morton, A somplete set of invariants for LU-equivalence of density operators, SIGMA \href{https://doi.org/10.3842/SIGMA.2017.028}{{\bf13}, 28-20 (2017).}

\bibitem{Maciazek2013}
T. Maci\c{a}\.{z}ek; M. Oszmaniec; A. Sawicki, How many invariant polynomials are needed to decide local unitary equivalence of qubit states? \jmp~\href{https://doi.org/10.1063/1.4819499}{{\bf54}, 092201 (2013).}

\bibitem{Viehmann2011}
O. Viehmann, C. Eltschka, J. Siewert, Polynomial invariants for discrimination and classification of four-qubit entanglement, \pra \href{https://doi.org/10.1103/PhysRevA.83.052330}{{\bf83}, 052330 (2011).}

\bibitem{Monras2011}
A. Monras, G. Adesso, S.M. Giampaolo, G. Gualdi, G.B. Davies, and F. Illuminati, Entanglement quantification by local unitary operations, \pra~\href{https://doi.org/10.1103/PhysRevA.84.012301}{{\bf84}, 012301 (2011).}

\bibitem{Wyderka2023}
N. Wyderka, A. Ketterer, S. Imai, J.L. B\"{o}nsel, D.E. Jones, B.T. Kirby, X-D. Yu, and O. G\"{u}hne, Complete characterization of quantum correlations by randomized measurements, \prl~\href{https://doi.org/10.1103/PhysRevLett.131.090201}{{\bf131}, 090201 (2023).}

\bibitem{Oszmaniec2024}
M. Oszmaniec, D.J. Brod and E.F. Galv\~{a}o, Measuring relational
information between quantum states, and applications, \njp
\href{https://doi.org/10.1088/1367-2630/ad1a27}{{\bf26}, 013053
(2024).}

\bibitem{Quek2024}
Y. Quek, E. Kaur, and M.M. Wilde, Multivariate trace estimation in
constant quantum depth, Quantum
\href{https://doi.org/10.22331/q-2024-01-10-1220}{{\bf8}, 1220
(2024).}

\bibitem{Procesi1976}
C. Procesi, The invariant theory of $n\times n$ matrices, Adv.Math.
\href{https://doi.org/10.1016/0001-8708(76)90027-X}{{\bf19}, 306
(1976).}

\bibitem{Bargmann1964}
V. Bargmann, Note on Wigner's theorem on symmetry operations, \jmp~\href{https://doi.org/10.1063/1.1704188}{{\bf5}, 862 (1964).}

\bibitem{Chien2016}
T.Y. Chien and S. Waldron, A characterization of projective unitary
equivalence of finite frames and applications, SIAM J. Discrete
Math. \href{https://doi.org/10.1137/15M1042140}{{\bf30}, 976
(2016).}

\bibitem{Mukunda2003}
N. Mukunda, P.K. Aravind, and R. Simon, Wigner rotations, Bargmann
invariants and geometric phases, \jpa: Math. Gen.
\href{https://doi.org/10.1088/0305-4470/36/9/312}{{\bf36}, 2347
(2003).}

\bibitem{Berry1984}
M.V. Berry, Quantal Phase Factors Accompanying Adiabatic Changes,
Proc. R. Soc. A
\href{https://doi.org/10.1098/rspa.1984.0023}{\textbf{392}, 45
(1984).}

\bibitem{Aharonov1987}
Y. Aharonov and J. Anandan, Phase change during a cyclic quantum
evolution, \prl
\href{https://doi.org/10.1103/PhysRevLett.58.1593}{\textbf{58}, 1593
(1987).}

\bibitem{Kirkwood1933}
J.G. Kirkwood, Quantum statistics of almost classical assemblies,
\href{https://doi.org/10.1103/PhysRev.44.31}{{\bf44}, 31(1933).}

\bibitem{Dirac1945}
P.A.M. Dirac, On the analogy between classical and quantum
mechanics, \rmp
\href{https://doi.org/10.1103/RevModPhys.17.195}{{\bf17},
195(1945).}

\bibitem{Schmid2024}
D. Schmid, R. D. Baldij\~{a}o, Y. Y\={\i}ng, R. Wagner, and J. H.
Selby, Kirkwood-Dirac representations beyond quantum states and
their relation to noncontextuality, \pra
\href{https://doi.org/10.1103/PhysRevA.110.052206}{\textbf{110},
052206 (2024).}

\bibitem{Buhrman2001}
H. Buhrman, R. Cleve, J. Watrous, and R. de Wolf, Quantum
fingerprinting, \prl \href{https://doi.org/10.1103/PhysRevLett.87.167902}{{\bf87}, 167902 (2001).}

\bibitem{Fernandes2024}
C. Fernandes, R. Wagner, L. Novo, and E.F. Galv\~{a}o,
Unitary-Invariant Witnesses of Quantum Imaginary, \prl
\href{https://doi.org/10.1103/PhysRevLett.133.190201}{{\bf133},
190201 (2024).}

\bibitem{Li2025}
M-S. Li, Y. Tan, Bargmann invariants for quantum imaginarity,
\href{https://doi.org/10.1103/PhysRevA.111.022409}{{\bf111}, 022409 (2025).}

\bibitem{Li2026}
M-S. Li, R. Wagner, and L. Zhang, Multi-state imaginarity and coherence in qubit systems,
\pra \href{https://doi.org/10.1103/10.1103/tpgw-v6ht}{{\bf113}, 012428 (2026).}

\bibitem{Wagner2025}
R. Wagner, Coherence and contextuality as quantum resources, Ph.D. Thesis
\href{http://arxiv.org/abs/2511.16785}{arXiv:2511.16785}

\bibitem{Zhang2025PRA1}
L. Zhang, B. Xie, and B. Li, Geometry of sets of Bargmann
invariants, \pra
\href{https://doi.org/10.1103/PhysRevA.111.042417}{{\bf111}, 042417
(2025).}

\bibitem{Zhang2025}
L. Zhang, B. Xie, and Y. Tao, Bargmann-invariant framework for local
unitary equivalence and entanglement, \pra~\href{https://doi.org/10.1103/s3mp-3kn6}{{\bf112}, 052426 (2025).}

\bibitem{Pratapsi2025}
S.S. Pratapsi, J. Gouveia, L. Novo, and E.F. Galv\~{a}o, Elementary
characterization of Bargmann invariants, \pra
\href{https://doi.org/10.1103/hsnv-wpt3}{{\bf112}, 042421 (2025).}

\bibitem{Xu2026}
J. Xu, Numerical ranges of Bargmann invariants,  \pla
\href{https://doi.org/10.1016/j.physleta.2025.131091}{{\bf565},
131091 (2026).}

\bibitem{Zhang2026}
L. Zhang and B. Xie, A Survey of Bargmann Invariants: Geometric Foundations and Applications, \href{http://arxiv.org/abs/2601.01858}{arXiv: 2601.01858}

\bibitem{Wagner2026a}
R. Wagner, Bargmann Scenarios, \href{https://arxiv.org/pdf/2604.18833}{arXiv:2604.18833}

\bibitem{Wagner2026b}
R. Wagner, E.F. Galv\~{a}o, Commutativity from a single Bargmann invariant equality, \href{https://arxiv.org/pdf/2605.07405}{arXiv:2605.07405}

\bibitem{Wang2026}
Y. Wang, A low order Bargmann invariant hierarchy for set coherence, \href{https://arxiv.org/abs/2605.10003}{arXiv:2605.10003}

\bibitem{Peres1996}
A. Peres, Separability Criterion for Density Matrices, \prl \href{https://doi.org/10.1103/PhysRevLett.77.1413}{{\bf77}, 1413(1996).}

\bibitem{Horodecki1996}
M. Horodecki, P. Horodecki, R. Horodecki, Separability of mixed states: necessary and sufficient conditions,  \pla \href{https://doi.org/10.1016/S0375-9601(96)00706-2}{{\bf223}(1-2), 1-8 (1996).}

\bibitem{Horn2012}
R.A. Horn and C.R. Johnson, Matrix Analysis, Cambridge University Press (2012).

\bibitem{Parthasarathy2005}
K.R. Parthasarathy, Extremal quantum states in coupled systems, Annales de l'I.H.P. Probabilités et statistiques, Tome \href{https://doi.org/10.1016/j.anihpb.2003.10.009}{{\bf41}(3), 257-268 (2005).}

\bibitem{Cerf1999}
N.J. Cerf, C. Adami, and R.M. Gingrich, Reduction criterion for separability, \pra \href{https://doi.org/10.1103/PhysRevA.60.898}{{\bf60},898 (1999).}

\bibitem{Horodecki1999}
M. Horodecki, P. Horodecki, Reduction criterion of separability and limits for a class of distillation protocols, \pra \href{https://doi.org/10.1103/PhysRevA.59.4206}{{\bf59}, 4206(1999).}

\bibitem{Gerdt2011}
V. Gerdt, D. Mladenov, Y. Palii, and A. Khvedelidze, SU(6) Casimir invariants and SU(2)$\ot$SU(3) scalars for a mixed qubit-qutrit state, J. Math. Sci. \href{https://doi.org/10.1007/s10958-011-0619-9}{{\bf179}, 690–701 (2011).} 

\bibitem{Macdonald1995}
I.G. Macdonald, Symmetric Functions and Hall Polynomials, 2nd Ed.,
Oxford University Press: Oxford, UK (1995).

\bibitem{Horodecki2000}
P. Horodecki, M. Lewenstein, G. Vidal, and I. Cirac, Operational criterion and constructive checks for the separability of low-rank density matrices,  \pra \href{https://doi.org/10.1103/PhysRevA.62.032310}{{\bf62}(3), 032310 (2000).}

\bibitem{Ma2026}
M. Ma and R. Shi, Bargmann invariants and local unitary equivalence, \href{http://arxiv.org/abs/2607.16878}{arXiv:2607.16878}

\bibitem{Mumford1994}
D. Mumford, J. Fogarty, and F. Kirwan, Geometric Invariant Theory, Springer-Verlag Berlin Heidelberg (1994).

\end{thebibliography}
\end{document}